\documentclass{lmcs}
\usepackage[utf8]{inputenc}
\pdfoutput=1

\usepackage{lastpage}
\lmcsdoi{19}{1}{1}
\lmcsheading{}{\pageref{LastPage}}{}{}%
{Dec.~13,~2016}{Jan.~05,~2023}{}

\keywords{Must-testing preorder; Multiparty composition; Mazurkiewicz traces}

\ACMCCS{[{\bf Theory of computation}]: Models of computation---Concurrency---Process Calculi}

\usepackage{hyperref}
\theoremstyle{plain} 

\usepackage{amssymb}
\usepackage{url}
\usepackage{amsthm}
\usepackage{amsfonts}
\usepackage{stmaryrd}
\usepackage{amsmath}
\usepackage{mathtools}
\usepackage{centernot}
\usepackage{color}
\usepackage{xargs}
\usepackage{booktabs}

\newcommand{\uncoordinated}{{individualistic}}
\newcommand{\Uncoordinated}{{Individualistic}}
\newcommand{\distributed}{{uncoordinated}}
\newcommand{\Distributed}{{Uncoordinated}}

\newcommand{\distributedobs}{{uncoordinated}}
\newcommand{\Distributedobs}{{Uncoordinated}}

\newcommand{\distributedtest}{{uncoordinated}}

\newcommand{\Filteredtr}{{Filtered traces}}
\newcommand{\filteredtr}{{filtered traces}}

\newcommand{\success}{\checkmark}

\newcommand{\observers}{\mathcal{O}}

\newcommand{\ie}{i.e.}

\newcommand{\bydef}{\stackrel{\scriptstyle{\mathit{def}}}{=}}

\newcommand{\recX}{\mathbf{rec}_X}

\newcommand{\cohide}{\upharpoonright}
\newcommand{\trc}[1]{\mathsf{str}({#1})}

\newcommand{\actset}{\mathsf{Act}}
\newcommand{\one}{\mathbf{1}}
\newcommand{\zero}{\mathbf{0}}

\newcommand{\irule}[2]{\frac{\textstyle\rule[-1.3ex]{0cm}{3ex}#1}%
{\textstyle\rule[-.5ex]{0cm}{3ex}#2}}

\def \anmathrule #1#2{
    \irule{#1}{#2}%
}

\newcommand{\n}{\mathtt{n}}

\newcommand{\tr}[1]{\stackrel{\ #1\ }{\rightarrowfill}}

\newcommand{\dtr}[1]{\stackrel{#1}{\Longrightarrow}}

\makeatletter
\def \rightarrowfill{\m@th\mathord{\smash-}\mkern-6mu%
  \cleaders\hbox{$\mkern-2mu\mathord{\smash-}\mkern-2mu$}\hfill
  \mkern-6mu\mathord\rightarrow}
\makeatother

\makeatletter
\def \Rightarrowfill{\m@th\mathord{\smash-}\mkern-6mu%
  \cleaders\hbox{$\mkern-2mu\mathord{\smash-}\mkern-2mu$}\hfill
  \mkern-6mu\mathord\Rightarrow}
\makeatother

\newcommand{\abst}[1]{\cohide{#1}}

\newcommandx{\tequiv}[1][1=D]{\equiv_{#1}}
\newcommandx{\tclass}[2][1=D]{[#2]_{#1}}
\newcommandx{\qm}[1][1=D]{\mathbb{M}({#1})}

\newcommand{\must}{{\sf must}}
\newcommand{\notmust}{\ \mathrlap{\ \ \ \not\ }{\sf must}}

\newcommand{\MUST}{{\sf MUST}}
\newcommand{\NOTMUST}{\ \mathrlap{\ \ \ \not\ }{\sf MUST}}

\newcommand{\mustleq}{\sqsubseteq_{\must}}
\newcommand{\musteq}{\approx_{\must}}

\newcommand{\bmustleq}{\preceq_{\must}}

\newcommand{\unc}{{\sf unc}}

\newcommandx{\tmustleq}[1][1=D]{\sqsubseteq_{\unc_{#1}}}
\newcommandx{\tmusteq}[1][1=D]{\approx_{\unc_{#1}}}

\newcommand{\dmustleq}[1]{\sqsubseteq_{\unc}^{#1}}
\newcommand{\dmusteq}[1]{\approx_{\unc}^{#1}}

\newcommandx{\btmustleq}[1][1=\mathbb{I}]{\preceq_{\unc}^{#1}}

\newcommand{\lmust}{{\sf ind}}
\newcommand{\lmustleq}[1]{\sqsubseteq_{\lmust}^{#1}}
\newcommand{\lmusteq}[1]{\approx_{\lmust}^{#1}}

\newcommand{\after}{\ {\sf after}\ }

\newcommand{\pref}[2]{\trc{#1,#2}}
\newcommand{\mpref}[2]{\pref{#1}{#2}}
\newcommandx{\tsplit}[3][3=\mathbb{I}]{#1^\dagger}

\newcommand{\co}[1]{\overline{#1}}

\newcommandx{\noisyeq}[1][1=I]{\stackrel{\bullet}{\equiv}_{#1}\ }
\newcommandx{\nclass}[2][1=I]{[[#2]]_{#1}}

\newcommandx{\bnmustleq}[1][1=I]{\preceq_{\lmust}^{#1}}

\newcommandx{\sdmust}[1][1=\mathcal{I}]{\preceq_{\unc}^{#1}}

\newcommand{\init}[1]{\mathsf{init}(#1)}

\usepackage{textcomp}
\usepackage{listings}
\usepackage{xcolor}

\makeatletter

\newcommand\PrologPredicateStyle{}
\newcommand\PrologVarStyle{}
\newcommand\PrologAnonymVarStyle{}
\newcommand\PrologAtomStyle{}
\newcommand\PrologOtherStyle{}
\newcommand\PrologCommentStyle{}

\newif\ifpredicate@prolog@
\newif\ifwithinparens@prolog@

\lst@SaveOutputDef{`_}\underscore@prolog

\newcount\currentchar@prolog

\newcommand\@testChar@prolog%
{%
  \ifnum\lst@mode=\lst@Pmode%
    \detectTypeAndHighlight@prolog%
  \else
    \ifwithinparens@prolog@%
      \detectTypeAndHighlight@prolog%
    \fi
  \fi
  \global\predicate@prolog@false%
}

\newcommand\detectTypeAndHighlight@prolog
{%
  \def\lst@thestyle{\PrologAtomStyle}%
  \ifpredicate@prolog@%
    \def\lst@thestyle{\PrologPredicateStyle}%
  \else
    \expandafter\splitfirstchar@prolog\expandafter{\the\lst@token}%
    \expandafter\ifx\@testChar@prolog\underscore@prolog%
      \ifnum\lst@length=1%
        \let\lst@thestyle\PrologAnonymVarStyle%
      \else
        \let\lst@thestyle\PrologVarStyle%
      \fi
    \else
      \currentchar@prolog=65
      \loop
        \expandafter\ifnum\expandafter`\@testChar@prolog=\currentchar@prolog%
          \let\lst@thestyle\PrologVarStyle%
          \let\iterate\relax
        \fi
        \advance \currentchar@prolog by 1
        \unless\ifnum\currentchar@prolog>90
      \repeat
    \fi
  \fi
}
\newcommand\splitfirstchar@prolog{}
\def\splitfirstchar@prolog#1{\@splitfirstchar@prolog#1\relax}
\newcommand\@splitfirstchar@prolog{}
\def\@splitfirstchar@prolog#1#2\relax{\def\@testChar@prolog{#1}}

\def\beginlstdelim#1#2%
{%
  \def\endlstdelim{\PrologOtherStyle #2\egroup}%
  {\PrologOtherStyle #1}%
  \global\predicate@prolog@false%
  \withinparens@prolog@true%
  \bgroup\aftergroup\endlstdelim%
}

\newcommand\lang@prolog{Prolog-pretty}
\expandafter\lst@NormedDef\expandafter\normlang@prolog%
  \expandafter{\lang@prolog}

\expandafter\expandafter\expandafter\lstdefinelanguage\expandafter%
{\lang@prolog}
{%
  language            = Prolog,
  keywords            = {},      
  showstringspaces    = false,
  alsoletter          = (,
  alsoother           = @\$,
  moredelim           = **[is][\beginlstdelim{(}{)}]{(}{)}, 
  MoreSelectCharTable =
    \lst@DefSaveDef{`(}\opparen@prolog{\global\predicate@prolog@true\opparen@prolog}, 
}

\newcommand\@ddedToOutput@prolog\relax
\lst@AddToHook{Output}{\@ddedToOutput@prolog}

\lst@AddToHook{PreInit}
{%
  \ifx\lst@language\normlang@prolog%
    \let\@ddedToOutput@prolog\@testChar@prolog%
  \fi
}

\lst@AddToHook{DeInit}{\renewcommand\@ddedToOutput@prolog{}}

\makeatother
\definecolor{PrologPredicate}{RGB}{000,031,255}
\definecolor{PrologVar}      {RGB}{024,021,125}
\definecolor{PrologAnonymVar}{RGB}{000,127,000}
\definecolor{PrologAtom}     {RGB}{186,032,032}
\definecolor{PrologComment}  {RGB}{063,128,127}
\definecolor{PrologOther}    {RGB}{000,000,000}
\definecolor{mygray}{gray}{0.5}
\definecolor{mylightgray}{gray}{0.9}

\renewcommand\PrologPredicateStyle{\color{PrologPredicate}}
\renewcommand\PrologVarStyle{\color{PrologVar}}
\renewcommand\PrologAnonymVarStyle{\color{PrologAnonymVar}}
\renewcommand\PrologAtomStyle{\color{PrologAtom}}
\renewcommand\PrologCommentStyle{\itshape\color{PrologComment}} 
\renewcommand\PrologOtherStyle{\color{PrologOther}}

\lstdefinestyle{Prolog-pygsty}
{
  language     = Prolog-pretty,
  upquote      = true,
  stringstyle  = \PrologAtomStyle,
  commentstyle = \PrologCommentStyle,
  literate     =
    {:-}{{\PrologOtherStyle :-}}2
    {,}{{\PrologOtherStyle ,}}1
    {.}{{\PrologOtherStyle .}}1
}

\lstdefinestyle{INLINE}{
}

\lstdefinestyle{DISPLAY}{
   numberstyle=\tiny\tt\color{mygray},
   backgroundcolor=\color{mylightgray},
   framerule=0pt,
   frame=tlbr, framesep=0.2cm
}
\lstnewenvironment{Prolog}[1][]{\lstset{language=Prolog,style=DISPLAY,style=FLOAT,#1}}{}
\lstnewenvironment{PrologD}[1][]{\lstset{language=Prolog-pretty,style=DISPLAY,numbers=left,#1}}{}
\lstnewenvironment{PrologN}[1][]{\lstset{language=Prolog-pretty,style=DISPLAY,#1}}{}
\newcommand{\PI}{\lstinline[language=Prolog-pretty]}
{}

\lstnewenvironment{PIC}[1][]{\lstset{language=PIC,style=DISPLAY,#1}}{}

\begin{document}

\title[Multiparty testing preorders]{Multiparty testing preorders}

\author[R. De Nicola]{Rocco De Nicola\lmcsorcid{0000-0003-4691-7570}}[a]	
\author[H. Melgratti]{Hern\'an Melgratti\lmcsorcid{0000-0003-0760-0618}}[b]	

\address{IMT  School for Advanced Studies Lucca, Italy}	
\email{rocco.denicola@imtlucca.it}  

\address{DC, FCEyN, Universidad de Buenos Aires - CONICET, Argentina}
\email{hmelgra@dc.uba.ar}  


\titlecomment{{\lsuper*}This is an extended and revised version of a paper
  with same title that appeared in the proceedings of the 10th International
  Symposium on Trustworthy Global Computing, Vol. 9533 of LNCS, Springer 2016}


\begin{abstract}
  Variants of the must testing approach have been successfully applied
  in service oriented computing for analysing the compliance between
  (contracts exposed by) clients and servers or, more generally,
  between two peers.  It has however been argued that multiparty
  scenarios call for more permissive notions of compliance because
  partners usually do not have full coordination capabilities. We
  propose two new testing preorders, which are obtained by restricting
  the set of potential observers. For the first preorder, called
  \distributed, we allow only sets of parallel observers that use
  different parts of the interface of a given service and have no
  possibility of intercommunication.  For the second preorder, that we
  call \uncoordinated, we instead rely on parallel observers that
  perceive as silent all the actions that are not in the interface of
  interest.  We have that the \distributed\ preorder is coarser than
  the classical must testing preorder and finer than the
  \uncoordinated\ one.  We also provide a characterisation in terms of
  decorated traces for both preorders: the \distributed\ preorder is
  defined in terms of must-sets and Mazurkiewicz traces while the
  \uncoordinated\ one is described in terms of classes of filtered
  traces that only contain designated visible actions and must-sets.
\end{abstract}

\maketitle


\section{Introduction}\label{sec:intro}

A desired property of communication-centered systems is the graceful
termination of the processes involved in a multiparty interaction, i.e., every
possible interaction ends successfully, in the sense that there are neither
messages waiting forever to be sent nor sent messages which are never
received.
The theories of session types~\cite{THK94,HondaVK98} and of
contracts~\cite{CastagnaGP08,CGP09:TCWS,BZ:TUTCCCC,LaneveP07} are commonly
used to ensure such kind of properties.
The key idea behind both approaches is to associate each process a with type
(or \textit{contract}) that gives an abstract description of its external,
visible behaviour and to use type checking to verify the correctness of
behaviours.

Processes are often defined as sequential nondeterministic {\sc ccs}
processes~\cite{Mil:CC} describing the offered communications, and are
built-up from send and receive actions. These activities are abstractly
represented either as output and input actions that take place over a set of
channels or as internal $\tau$ actions.
Basic actions can be composed sequentially (prefix operator ``.'') or as
alternatives (non deterministic choice ``$+$'').
Typically, the language for describing processes does not have any operator
for parallel composition. It is assumed that all possible interleavings are
made explicit in the description of the service and communication is only used
for modelling the interaction among different processes.

In \emph{client-server} scenarios, i.e., in settings involving just two
processes, variants of the must testing preorder has been used to compare
alternative implementations of servers and
clients~\cite{bernardi2013mutually}.
Technically, two given processes $p$ and $q$ are related via the must preorder
($p\ \mustleq q$) if $q$ satisfies all observers that are satisfied by $p$.
Consequently, $p$ and $q$ are considered equivalent ($p\ \musteq q$) if they
satisfy exactly the same observers.
Standardly, an observer is a unique (sequential) process that runs in parallel
with the tested process and, consequently, all interactions with the tested
process are handled by a unique, central process, i.e., the observer itself.

If one considers a multiparty setting, each process may concurrently interact
with several other partners and its interface is often (logically) partitioned
by allowing each partner to communicate only through a dedicated part of the
interface.
Consider the following scenario involving three partners: an organisation (the
broker) that sells goods produced by a different company (the producer) to a
specific customer (the client). The behaviour of the broker can be described
by the following process:
\[B = \mathit{req}.\co{\mathit{ order}}.\co{\mathit{inv}}.\zero\]
The broker accepts requests on channel $\mathit{req}$ and then places an order to
the producer with the message $\co{\mathit{order}}$ and sends an invoice to the
client with the message $\co{\mathit{inv}}$.
A client may behave as the process $C$ below, which first sends a request on
channel $\mathit{req}$ and then expects the invoice on channel $\mathit{inv}$, i.e.,
\[C = \co{\mathit{req}}.\mathit{inv}.\zero\]
A producer may be modelled by a process that simply accepts an order over
channel $\mathit{ord}$, i.e.,
\[P = \mathit{order}.\zero\]
In this scenario, the broker uses the channels $\mathit{req}$ and
$\mathit{inv}$ to interact with the client, while it interacts with the
producer over the channel $\mathit{order}$.
Moreover, the client and the producer do not know each other and are
completely independent. Hence, the order in which messages $\co{\mathit{order}}$
and $\co{\mathit{inv}}$ are sent is completely irrelevant for them.
In fact, they would be equally happy with a broker defined as follows:
\[B'= \mathit{req}.\co{\mathit{inv}}.\co{\mathit{order}}.\zero\]
Nevertheless, these two different implementations are not considered
must-equivalent.
In these situations, the classical must testing preorder turns out to be
unnecessarily discriminating.

The main goal of this paper is to introduce alternative, less discriminating,
preorders that take into account the distributed nature of the observers and
their possibly limited coordination and interaction capabilities.
A first preorder, called \emph{\distributed\ must preorder}, is obtained by
assuming that a set of observers of a given process interact with it via fully
disjoint sets of ports, i.e., they use different parts of its interface, have
no possibility of intercommunication, and all of them terminate successfully
in every possible interaction. It is however worth noting that these
assumptions about the absence of communication among observers do not fully
eliminate the possibility for observers of being mutually influenced, e.g.,
when one observer does not enable a communication on some ports.
Due to this, it is possible to differentiate $B$ from $B'$ above when one of
the observers refuses to synchronise over a port, e.g., if the client does not
enable the synchronisation over the channel $\mathit{inv}$. Consider a client
$ C' = \co{\mathit{req}}.\zero$ that sends a request and terminates without
accepting the invoice. While $P$ and $C'$ always terminate their interaction
with $B$, this is not the case when interacting with $B'$ because
communication over channel $\mathit{order}$ is never enabled. Consequently, these
two implementations of the broker are distinguished by the \distributed\ must
preorder. However, the \distributed\ must preorder allows for the reordering
of actions. For instance, the following two implementations of the broker are
considered equivalent under the \distributed\ must preorder.
\[
  \begin{array}{l@{\ =\ }l}
    B'' & \mathit{req}.(\co{\mathit{order}}.\co{\mathit{inv}}.\zero
          + \co{\mathit{order}}.\zero + \co{\mathit{inv}}.\zero)\\
    B''' & \mathit{req}.(\co{\mathit{inv}}.\co{\mathit{order}}.\zero + \co{\mathit{order}}.\zero + \co{\mathit{inv}}.\zero)
  \end{array}
\]
We remark that the processes $C'$ and $P$ are not able to distinguish $B''$
from $B'''$, because they both terminate when interacting with either $B''$
and $B'''$. We also note that a client behaving as described by $C$ will not
be satisfied neither by $B''$ nor $B'''$ because they both may decide not to
communicate over $\mathit{inv}$.

The second preorder, which we call \emph{\uncoordinated\ must preorder},
allows observers to take for granted the execution of those actions of the
process that are not explicitly of interest for them (i.e., not in their
alphabet). For instance, a client in the previous scenario assumes that the
producer will always enable the communication over the channel $\mathit{order}$.
In the {\uncoordinated\ must preorder}, the processes $B$ and $B'$ turn out to
be indistinguishable.

The preorders are, as usual, defined in terms of the outcomes of experiments
by specific sets of observers. For defining the \distributed\ must preorder,
we allow only sets of parallel observers that cannot intercommunicate and do
challenge processes via disjoint parts of their interface. For defining the
\uncoordinated\ must preorder, we instead rely on parallel observers that,
again, cannot intercommunicate but in addition perceive as silent all the
actions that are not part of the interface of their interest. This is
instrumental to avoid that a specific observer recovers information about
other involved observers. As expected, we have that the \distributed\ preorder
is coarser than the classical must testing preorder and finer than the
\uncoordinated\ one.

Just like for classical testing preorders, we provide a characterisation for
both new preorders in terms of decorated traces, which avoids dealing with
universal quantifications over the set of observers whenever a specific
relation between two processes has to be established. The alternative
characterisations make it even more evident that our preorders permit action
reordering. Indeed, the \distributed\ preorder is defined in terms of
\emph{Mazurkiewicz traces}~\cite{mazurkiewicz1995introduction} while the
\uncoordinated\ one is described in terms of classes of traces quotiented via
specific sets of visible actions. We would like to remark that our preorders
are different from those defined
in~\cite{bravetti2008foundational,padovani2010contract,mostrous2009global},
which also permit action reordering by relying on buffered communication;
additional details will be provided in Section~\ref{sec:conclusions}.

{\bf Synopsis} The remainder of this paper is organised as follows. In
Section~\ref{sec:centralised} we recall the basics of the classical must testing
approach.
In Section~\ref{sec:mazurkiewicz} and Section~\ref{sec:uncoordinated} we present the
theory of \distributed\ and \uncoordinated\ must testing preorders and their
characterisation in terms of traces. In Section~\ref{sec:relation} we show that the
\distributed\ preorder is coarser than the must testing preorder but finer
than the \uncoordinated\ one.
In Section~\ref{atwork} we describe a Prolog implementation of the \distributed\
and \uncoordinated\ preorders for the finite fragment of our specification
language and use it for analysing a scenario involving a replicated data
store.
Finally, we discuss some related work and future developments in
Section~\ref{sec:conclusions}.

This paper is a revised and extended version of~\cite{NicolaM15}. We fix an
incorrect characterisation of the \distributed\ preorder~\cite{NicolaM15}, and
provide full proofs of previously published results. In addition, we give a
prototype implementation in Prolog of the alternative characterisation of the
proposed preorders for the fragment of the calculus with only finite processes
and illustrate the usability of the proposed preorders by using them to reason
on different implementations of components in a replicated data store
(Section~\ref{atwork}).



\section{Processes and testing preorders}\label{sec:centralised}

Let $\mathcal{N}$ be a countable set of action names, ranged over by
$a, b, \ldots$. As usual, we write co-names in $\co{\mathcal{N}}$ as
$\co{a},\co{b}, \ldots$ and assume $\co{\co{a}}=a$. We will use $\alpha$,
$\beta$ to range over $\actset =(\mathcal{N}\cup\co{\mathcal{N}})$. Moreover,
we consider a distinguished internal action $\tau$ not in $\actset$ and use
$\mu$ to range over $\actset\cup\{\tau\}$. We fix the language for defining
processes as the sequential fragment of {\sc ccs} extended with a
\emph{success} operator, as specified by the following grammar.
\begin{center}
  \begin{math}
    \begin{array}{l@{\ ::=\ }ll}
      p, q & \zero\ |\ \one\ |\ \mu.p \ |\ p+q \ | \ X \ | \ \recX. p
    \end{array}
  \end{math}
\end{center}

The process $\zero$ stands for the terminated process, $\one$ for the process
that reports success and then terminates, and $\mu.p$ for a process that
executes $\mu$ and then continues as $p$. Alternative behaviours are specified
by terms of the form $p + q$, while recursive ones are introduced by terms
like $\recX.p$. We denote by $\mathcal{P}$ the set of all processes.
We write $\n(p)$ for the set of names $a\in\mathcal{N}$ such that either $a$
or $\co{a}$ occur in $p$.

The operational semantics of processes is given in terms of a labelled
transition system (\textsc{lts}) $p\tr{\lambda} q$ with
$\lambda\in\actset\cup\{\tau,\success\}$, where $\success$ signals the
successful termination of an execution.

\begin{defi}[Transition relation]\label{def:lts-processes}
  The transition relation on processes, noted $\tr{\lambda}$, is the
  least relation satisfying the following rules
  \begin{center}
    \begin{math}
      \begin{array}{l@{\qquad}l@{\qquad}l@{\qquad}l@{\qquad}l}
        \one \tr{\success} \zero
        &
        \mu.p \tr{\mu} p
        &
        \anmathrule{p\tr{\lambda} p'}
                   {p+q\tr{\lambda} p'}
                   &
                   \anmathrule{q\tr{\lambda} q'}
                              {p+q\tr{\lambda} q'}
                              &
                              \anmathrule{p[\recX. p / X] \tr{\lambda} p'}
                                         {\recX. p \tr{\lambda} p'}
      \end{array}
    \end{math}
  \end{center}

\end{defi}

Multiparty applications, named \emph{configurations}, are built by composing
processes concurrently. Formally, configurations are given by the following
grammar.

\begin{center}
  \begin{math}
    c,d, o::= p \ |\  c\| d
  \end{math}
\end{center}

We denote by $\observers$ the set of all configurations. We sometimes write
$\Pi_{i\in 0..n} p_i$ for the parallel composition
$p_0 \ \|\ \ldots\ \|\ p_n$.
The operational semantics of configurations, which accounts for the
communication between processes, is obtained by extending the {\sc lts} in
Definition~\ref{def:lts-processes} with the following rules:
\begin{center}
  \begin{math}
    \anmathrule{c\tr{\mu}c'}{c \ \|\  d \tr{\mu} c' \  \|\  d}
    \qquad
    \anmathrule{d\tr{\mu}d'}{c \ \|\  d \tr{\mu} c \ \|\  d'}
    \qquad
    \anmathrule{c\tr{\alpha}c'\quad d\tr{\co{\alpha}}d'}{c\ \|\ d \tr{\tau} c' \ \|\ d'}
    \qquad
    \anmathrule{c\tr{\success}c'\quad d\tr{{\success}}d'}{c\ \|\ d \tr{\success} c' \ \|\ d'}
  \end{math}
\end{center}

All rules are standard apart for the last one that is not present
in~\cite{DBLP:journals/tcs/NicolaH84}. This rule states that the
concurrent composition of processes can report success only when all
processes in the composition do so.

We write $c\tr{\lambda}$ when there exists $c'$ such that $c\tr{\lambda}c'$;
$\Rightarrow$ for the reflexive and transitive closure of $\tr{\tau}$;
$c\dtr{\lambda}c'$ for $\lambda\in\actset\cup\{\success\}$ and
$c\dtr{}\tr{\lambda}\dtr{}$; $c\dtr{\lambda_0\ldots\lambda_n} c'$ for
$c\dtr{\lambda_0}\ldots\dtr{\lambda_n} c'$, and $c\dtr{s}$ with
$s\in{(\actset\cup\{\success\})}^*$ if there exists $c'$ such that $c\dtr{s}c'$.
We say that $c_o\tr{\mu_0} \ldots c_i\tr{\mu_i} \ldots$ is successful when
there exists $j $ s.t. $c_j\tr{\success}$; it is unsuccessful otherwise.

We write $\trc c$ and $\init c $ to denote the sets of strings and enabled
actions of $c$, defined as follows
\begin{center}
  \begin{math}
    \trc c = \{s\in{(\actset\cup\{\success\})}^*\ |\ c\dtr s\}
    \qquad \qquad
    \init c = \{ \lambda\in\actset\cup\{\success\} \ | \ c\dtr{\lambda}\}
  \end{math}
\end{center}

As behavioural semantics, we consider the must-testing
preorder~\cite{DBLP:journals/tcs/NicolaH84}, which is defined in terms of the
computations of a process under test $p$ and an observer $o$. A computation of
$p\ \|\ o$ is a sequence of $\tau$ transitions, i.e.,
\[p\ \|\ o = p_0\ \|\ o_0\tr{\tau}\ldots \tr{\tau}p_k\ \|\ o_k\tr{\tau}
  \ldots\]

A computation is \emph{maximal} if it is either infinite or its last term
$p_n\ \|\ o_n$ is such that $p_n\ \|\ o_n\not\tr{\tau}$.
We say it is \emph{observer-successful} if there exists $j\ge 0$ such
that $o_j\tr{\success}$, and \emph{observer-unsuccessful} otherwise.

\begin{defi}[must]
  We write $p\ \must\ o$ iff for each maximal computation of $p\ \|\ o$ is
  observer-successful.
\end{defi}

The notion of passing a test (or satisfying an observer) represents the fact
that an observer built-up from the parallel composition of processes is able
to report success in every possible interaction with the process under test.
It is then natural to compare processes according to their capacity to satisfy
observers.

The standard framework of~\cite{DBLP:journals/tcs/NicolaH84} can be recovered
by considering only observers without parallel composition.

\begin{defi}[must preorder]
  $p \mustleq q$ iff $\forall r\in\mathcal{P}: p\ \must\ r$ implies
  $q\ \must\ r$. We write $p\musteq q$ when both $p \mustleq q$ and
  $q \mustleq p$.
\end{defi}

\subsection{Semantic characterisation}
The must testing preorder has been characterised in terms of (i) the sequences
of actions that a process may perform, and (ii) the possible sets of actions
that it may perform after executing a particular sequence of
actions~\cite{DBLP:journals/tcs/NicolaH84}. This characterisation relies on a
few auxiliary predicates and functions that are presented below. A process $p$
\emph{diverges}, written $p\Uparrow$, when it exhibits an infinite, internal
computation $p\tr{\tau} p_0\tr{\tau} p_1\tr{\tau} \ldots $. We say $p$
\emph{converges}, written $p\Downarrow$, otherwise. For $s\in\actset^*$, the
convergence predicate is inductively defined by the following rules:
\begin{itemize}
\item
  $p \Downarrow \epsilon$ if $p\Downarrow$.
\item
  $p \Downarrow \alpha.s$ if $p\Downarrow$ and $p\dtr{\alpha} p'$
  implies $p'\Downarrow s$.
\end{itemize}

\noindent
The \emph{residuals} of a process $p$ (or a set of processes $P$) after the
execution of $s\in\actset^*$ is given by the following equations
\begin{itemize}
\item
  $ (p \after s) = \{p'\ |\ p\dtr{s} p'\}$.
\item
  $ (P \after s) = \{ p' \ |\ p\in P, \ p'\in (p \after s)\}$.
\end{itemize}

\begin{defi}[Must-set]\label{def:must-set}
  A \emph{must-set} of a process $p$ (or set of processes $P$) is
  $L\subseteq\actset$, and $L$ {finite} such that
  \begin{itemize}
  \item $p\ \MUST\ L$ iff $\forall p'$ such that $p\dtr{} p'$,
    $\exists \alpha\in L$ such that $p'\dtr{\alpha}$.
  \item $P\ \MUST\ L$ iff $\forall p\in P. p\ \MUST\ L$.
  \end{itemize}
\end{defi}

\noindent
Then, the must testing preorder can be characterised in terms of
strings and must-sets as follows.

\begin{defi}\label{def-charact-must} $p\bmustleq q$ if for every
  $s\in\actset^*$, for all finite $L\subseteq \actset$, if $p\Downarrow s$
  then
  \begin{itemize}
  \item $q\Downarrow s$.
  \item $(p \after { s})\ \MUST\ L$ implies $(q \after { s})\ \MUST\ L$.
  \end{itemize}
\end{defi}

\begin{thmC}[{\cite[Theorem 6.4.5]{DBLP:journals/tcs/NicolaH84}}]
  ${\mustleq} = {\bmustleq}$.
\end{thmC}


\section{A testing preorder with \distributed\
  observers}\label{sec:mazurkiewicz}

The must testing preorder is defined in terms of the tests that each process
is able to pass. Remarkably, the classical setting can be formulated by
considering only sequential tests (see the characterisation of minimal tests
in~\cite{DBLP:journals/tcs/NicolaH84}).
Each sequential test is a unique, centralised process that handles all the
interaction with the process under test and, therefore, has a complete view of
the externally observable behaviour of the process. For this reason, we refer
to the classical must testing preorder as a \emph{centralised preorder}.

Multiparty interactions are generally structured in such a way that pairs of
partners communicate through dedicated channels, for example, when relying on
partner links in service oriented models or buffers in communicating
machines~\cite{basu2012deciding}.
Conceptually, the interface (i.e., the set of channels) of a process is
partitioned and the process interacts with each partner by using only specific
sets of channels in its interface.
In addition, there are scenarios (as the one discussed in
Section~\ref{sec:intro}) in which partners frequently do not know each other
and cannot communicate directly.
As a direct consequence, the partners of a process cannot establish causal
dependencies among actions that take place over different parts of the
interface.
These constraints reduce the discriminating power of partners and call for
equivalences that equate processes that cannot be distinguished by sets of
independent sequential processes.

\begin{exa}\label{ex:trip}
  Consider the classical scenario for planning a trip. A user $U$ interacts
  with a broker $B$, which is responsible for booking flights provided by a
  service $F$ and hotel rooms available at service $H$. The expected
  interaction can be described as follows: $U$ makes a booking request by
  sending a message $req$ to $B$ (we will just describe the interaction and
  abstract away from data details such as trip destination, departure dates
  and duration). Depending on the request, $B$ may internally decide to
  contact service $F$ (for booking just a flight ticket), $H$ (for booking
  rooms) or both. Service $B$ uses channels $reqF$ and $reqH$ to respectively
  contact $F$ and $H$ (for the sake of simplicity, we assume that any request
  to $F$ and $H$ will be granted). Then, the expected behaviour of $B$ can be
  described with the following process. (As usual, $\tau$ actions and $+$ are
  combined to model internal, non-deterministic choices in a process.)

  \begin{center}
    \begin{math}
      B_0 \bydef \mathit{req}.(\tau.\co{\mathit{reqF}}.\zero
      +
      \tau.\co{\mathit{reqH}}.\zero
      +
      \tau.\co{\mathit{reqH}}.\co{\mathit{reqF}}.\zero)
    \end{math}
  \end{center}

  In this process, the third branch represents $B$'s choice to contact first
  $H$ and then $F$. Nevertheless, the other partners ($U$, $F$ and $H$) are
  not affected in any way by this choice, thus they would be equally happy
  with alternative definitions such as:

  \begin{center}
    \begin{math}
      \begin{array}{l@{\bydef} l}
        B_1
        &
        \mathit{req}.(\tau.\co{\mathit{reqF}}.\zero
        + \tau.\co{\mathit{reqH}.\zero
        + \tau.\co{\mathit{reqF}}}.\co{\mathit{reqH}}.\zero)
        \\
        B_2
        &
        \mathit{req}.(\tau.\co{\mathit{reqF}}.\zero
        + \tau.\co{\mathit{reqH}}.\zero
        + \tau.\co{\mathit{reqH}}.\co{\mathit{reqF}}.\zero
        +  \tau.\co{\mathit{reqF}}.\co{\mathit{reqH}}.\zero)
      \end{array}
    \end{math}
  \end{center}

  $B_0$, $B_1$ and $B_2$ are distinguished by the must testing equivalence. It
  suffices to consider
  $o_0= \co{\mathit{req}}.(\tau.\one + \mathit{reqF}.(\tau.\one +
  \mathit{reqH}. \zero))$ for showing that $B_0\not\mustleq B_1$ and that
  $B_0\not\mustleq B_2$, and use
  $o_1= \co{\mathit{req}}.(\tau.\one + \mathit{reqH}.(\tau.\one +
  \mathit{reqF}.\zero))$ for proving that $B_1\not\mustleq B_2$.
  \qed%
\end{exa}

This section is devoted to the definition and characterisation of a preorder,
called \emph{\distributed\ must preorder}, that is coarser than the classical
must preorder and relates processes that cannot be distinguished by
distributed contexts.
The \distributed\ must preorder is obtained by restricting the set of
observers to parallel processes that do not share channels.
We will say $\mathbb{I} = {\{I_i\}}_{i\in 0 \ldots n}$ is an {interface}
whenever $\mathbb{I}$ is a partition of $\actset$ and
$\forall \alpha\in\actset$, $\alpha\in I_i$ implies $\co{\alpha}\in I_i$. In
the rest of this paper we will usually write only the relevant part of an
interface.
For instance, we will write $\{\{a\},\{b\}\}$ for any interface $\{I_0,I_1\}$
such that $a\in I_0$ and $b\in I_1$. Then, the observers used by the
\distributed\ must testing preorder are introduced by the following
definition.

\begin{defi} [\Distributedobs\ observer]
  Let $\mathbb{I}= {\{I_i\}}_{i\in 0 \ldots n}$ be an interface. A process
  $\Pi_{i\in 0..n} o_i = o_0\ \|\ \ldots\ \|\ o_n$ is an
  \emph{\distributedobs\ observer over $\mathbb{I}$} if
  $\n(o_i) \subseteq I_i$ for all $i\in 0 \ldots n$.
\end{defi}

We say $o$ is an \distributedobs\ observer and omit the interface when no
confusion arises.
In our setting, which does not involve name mobility, the fact that
$\mathbb{I}= \{I_i\}$ is a partition of $\actset$ and $\n(o_i) \subseteq I_i$
suffices to avoid a direct communication among the processes of an
\distributedtest\ observer.
As a consequence, a distributed observer cannot impose a total order between
actions that are controlled by different processes of the observer.
Indeed, the executions of a distributed observer are the interleavings of the
executions of all processes ${\{o_i\}}_{i\in 0..n}$ (this property is formally
stated in Section~\ref{subsec:marzukiewicz},
Lemma~\ref{aux-lem:commutative-observer-behavior}).
We remark that a configuration does report success (i.e., perform action
$\success$) only when all processes in the composition do report success;
consequently an \distributedtest\ observer reports success when all its
components report success simultaneously.
Our definition of success deviates from the original setup
of~\cite{DBLP:journals/tcs/NicolaH84}. If success were not synchronised, e.g.,
every process would pass the observer $o = a.\zero\ ||\ \one$ because $o$
would be able to report success immediately. This is not the case in our
setting. In fact, each component of an \distributedtest\ observer accounts for
the view that a particular partner has about the process under test, and we
expect every component of the observer to be able to report success when a
process passes a test.

\begin{defi}[\Distributed\ must preorder $\dmustleq{\mathbb{I}}$]
  Let $\mathbb{I} = {\{I_i\}}_{i\in 0 \ldots n}$ be an interface. We say
  $p\dmustleq{\mathbb{I}} q$ iff for all \distributedobs\ observer $o$ over
  $\mathbb{I}$, $p\ \must\ o$ implies $q\ \must\ o$. We write
  $p\dmusteq{\mathbb{I}} q$ when both $p \dmustleq{\mathbb{I}} q$ and
  $q \dmustleq{\mathbb{I}} p$.
\end{defi}

\begin{exa}\label{ex:trip-dmust-bis} Consider the scenario presented in
  Example~\ref{ex:trip} and the following interface
  $\mathbb{I}=\{\{\mathit{req}\}, \{\mathit{reqF}\},\{\mathit{reqH}\}\}$ for
  the process $B$ that thus interacts with each of the other partners by using
  a dedicated part of its interface. It can be shown that the three
  definitions for $B$ in Example~\ref{ex:trip} are equivalent when considering
  the \distributed\ must testing preorder, i.e.,
  $B_0 \dmusteq{\mathbb{I}} B_1 \dmusteq{\mathbb{I}} B_2$. The actual proof,
  which uses the (trace-based) alternative characterization of the preorder,
  is deferred to Example~\ref{ex:trip-dmust}. \qed%
\end{exa}

\subsection{Semantic characterisation}\label{subsec:marzukiewicz}

We now characterise the \distributed\ must testing preorder in terms of traces
and must-sets. In order to do that, we shift from strings to Mazurkiewicz
traces~\cite{DBLP:conf/ac/Mazurkiewicz86}.
A \emph{Mazurkiewicz trace} is a set of strings, obtained by permuting
independent symbols. Traces represent concurrent computations, in which
commuting symbols stand for actions that execute independently of one another
and non-commuting symbols are causally dependent actions. We start by
summarising the basics of the theory of traces
in~\cite{DBLP:conf/ac/Mazurkiewicz86}.

Let $D \subseteq \actset\times\actset$ be a finite equivalence relation,
called the \emph{dependency relation}, that relates the actions that cannot be
commuted. Thus, if ${(\alpha,\beta)}\in D$ then $\alpha$ and $\beta$ are
considered causally dependent.
Symmetrically, $I_D = (\actset\times\actset) \setminus D$ stands for the
\emph{independency} relation with $({\alpha},{\beta})\in I_D$ meaning that
${\alpha}$ and ${\beta}$ are concurrent.

The trace equivalence induced by the dependency relation $D$ is the least
congruence $\equiv_D$ on $\actset^{*}$ such that for all
${\alpha,\beta}\in\actset: ({\alpha,\beta})\in I_D \implies {\alpha\beta}
\tequiv {\beta\alpha}$.

The equivalence classes of $\equiv_D$, denoted by $\tclass{s}$, are the
Mazurkiewicz \emph{traces}, namely the strings quotiented via $\equiv_D$.
We remark that no action can commute with $\success$ because $I_D$ is defined
over $\actset\times\actset$.

Let $\mathbb{I}$ be an interface, the \emph{dependency relation induced by}
$\mathbb{I}$ is $D = \bigcup_{I\in\mathbb{I}} I\times I$.

\begin{exa}
  Consider the interface
  $\mathbb{I}=\{\{\mathit{req}\}, \{\mathit{reqF}\}, \{\mathit{reqH}\}\}$ in
  Example~\ref{ex:trip-dmust-bis}.
  We recall that
  $\mathbb{I}=\{\{\mathit{req}\}, \{\mathit{reqF}\}, \{\mathit{reqH}\}\}$ is a
  convenient notation for any partition of $\actset$ such that
  $\forall \alpha\in\actset$, $\alpha\in I_i$ implies $\co{\alpha}\in I_i$.
  Consequently,
  \[
    \mathbb{I}=\{
    \{\mathit{req},  \co{\mathit{req}},  \ldots\},
    \{\mathit{reqF}, \co{\mathit{reqF}}, \ldots\},
    \{\mathit{reqH}, \co{\mathit{reqH}}, \ldots\}\}
  \]
  The dependency relation induced by $\mathbb{I}$ is
  \[
    D = \{
    (\mathit{req}, \mathit{req}),
    (\mathit{req}, \co{\mathit{req}}),
    (\co{\mathit{req}}, \mathit{req}),
    (\co{\mathit{req}}, \co{\mathit{req}}),
    \ldots,
    (\mathit{reqF},\mathit{reqF}),
    \ldots,
    (\mathit{reqH},\mathit{reqH}),
    \ldots
    \}
  \]
  Then, $(\alpha,\beta)\in I_D$ iff $(\alpha,\beta)\not\in D$. The relation
  $I_D$ basically states that actions that take place over channels that
  belong to different parts of the interface can commute.
  For instance, $\mathit{req}$, $\mathit{reqH}$ and $\mathit{reqF}$ are
  independent, and hence
  \[
    \mathit{req}\,\mathit{reqH}\,\mathit{reqF}
    \equiv_D
    \mathit{req}\,\mathit{reqF}\,\mathit{Reqh}
    \equiv_D
    \mathit{reqH}\,\mathit{reqF}\,\mathit{req}
    \equiv_D\ldots
  \]
  Consequently,
  \[
    \tclass{\mathit{req}\,\mathit{reqH}\,\mathit{reqF} }
    =
    \{
    \mathit{req}\,\mathit{reqH}\,\mathit{reqF}, \
    \mathit{req}\,\mathit{reqF}\,\mathit{reqH}, \
    \mathit{reqH}\,\mathit{reqF}\,\mathit{req}, \
    \ldots
    \}
  \]
  On the contrary, actions that take place over channels in the same part of
  the interface are dependent and cannot commute. Hence,
  $\mathit{req}\ \co{\mathit{req}} \not\equiv_D \co{\mathit{req}}\ {\mathit{req}}$.
  \qed%
\end{exa}

Now we are able to characterise the behaviour of an \distributedobs\ observer.
We start by formally stating that an \distributedobs\ observer reaches the
same configuration after executing any of the strings in the same equivalence
class. This result is instrumental to the proof of the alternative
characterisation of the \distributed\ preorder.

\begin{lem}\label{aux-lem:commutative-observer-behavior} Let
  $o = \Pi_{i \in 0..n} o_i$ be an \distributedobs\ observer over
  $\mathbb{I} = {\{I_i\}}_{i\in 0..n}$ and $D$ the dependency relation induced
  by $\mathbb{I}$. Then, for all $s\in\actset^*$ and $t\in\tclass s$ we have
  $o\dtr{s} o'$ iff $o\dtr{t} o'$.
\end{lem}

\begin{proof}
  The proof follows by induction on the length $|s|=|t|=n$.

  \begin{itemize}
  \item ${\bf n = 0,1}$. Immediate, because $s = t$.
  \item ${\bf n > 1}$. By the Levi's Lemma for traces~\cite[{Theorem
      1}]{mazurkiewicz1995introduction}, any possible choice of $v,w,x,y$ such
    that $s=vw$ and $t=xy$, implies that $v\tequiv z_1z_2$, $w\tequiv z_3z_4$,
    $x\tequiv z_1z_3$, $y\tequiv z_2z_4$ with $(z_2,z_3)\in I_D$. Consider a
    decomposition such that $|v|,|w|,|x|,|y|>0$ (this is always possible,
    because $n>1$). By inductive hypothesis on the reductions $o\dtr{v}o''$
    and $o'' \dtr{w} o'$, we have $o\dtr{v}o'' \dtr{w} o'$ iff
    $o\dtr{z_1}o_1\dtr{z_2}o_2\dtr{z_3}o_3\dtr{z_4}o'$. Since $z_2$ and $z_3$
    are independent they take part on different components of the observer and
    $o_1\dtr{z_2}o_2\dtr{z_3}o_3$ iff $o_1\dtr{z_3}o_2'\dtr{z_2}o_3$.
    Consequently, $o\dtr{v}\dtr{w} o'$ iff
    $o\dtr{z_1}o_1\dtr{z_2}o_2\dtr{z_3}o_3\dtr{z_4}o'$ iff
    $o\dtr{z_1} o_1\dtr{z_3}o_2'\dtr{z_2}o_3\dtr{z_4}o'$. By inductive
    hypothesis on reductions $o\dtr{z_1z_3} o_2'$ and $o_2'\dtr{z_2z_4}o'$ we
    have $o\dtr{z_1}o_1\dtr{z_2}o_2'\dtr{z_3}o_3\dtr{z_4}o'$ iff
    $o\dtr{x}\dtr{y} o'$.
    \qedhere
  \end{itemize}
\end{proof}

\begin{lem}\label{aux-cor:commutative-observer-behavior} Let
  $o = \Pi_{i \in 0..n} o_i$ be an \distributedobs\ observer over
  $\mathbb{I} = {\{I_i\}}_{i\in 0..n}$ and $D$ the dependency relation induced
  by $\mathbb{I}$. Then, $\forall s\in\actset^*, t\in\tclass s$,
  \begin{enumerate}
  \item\label{aux-lem:comm-observer-generates-whole-trace}
    $s\in\trc o$ implies $t\in\trc{o}$.
  \item\label{aux-lem:comm-observer-preserve-divergence}
    $o\Downarrow s$ implies $o\Downarrow t$.
  \item\label{aux-lem:comm-observer-refusal-sets}
    $(o \after s)\ \MUST\ L $ implies $(o \after t)\ \MUST\ L $.
  \item\label{aux-lem:comm-observer-unsuccessful-comp}
    If there exists an unsuccessful computation $o\dtr{s}$, then there
    exists an unsuccessful computation $o\dtr{t}$.
  \end{enumerate}
\end{lem}

\begin{proof}
  All items follow from Lemma~\ref{aux-lem:commutative-observer-behavior}. We
  illustrate (2). By contradiction. Assume $o\Downarrow s$ and $o\Uparrow t$.
  Then, there exist $t_1,t_2$ such that $t =t_1t_2$, $o \dtr{t_1} o'$ and
  $o'\Uparrow$. Without loss of generality, we assume that $o \dtr{t_1} o'$ is
  minimal, i.e., $o=o_0 \tr{\mu_1} o_1 \ldots \tr{\mu_{k}} o_k = o'$ with
  $t_1 = \mu_1\ldots\mu_k$ and $\forall i < k.o_i \Downarrow$. Since
  $o' = \Pi_{i \in 0..n} r_i$, there exists $h \in{0..n}$ s.t. $r_h\Uparrow$.
  Then, take $s_1$ and $s_2$ such that $s = s_1s_2$ and
  $s_1\cohide I_h = t_1\cohide I_h$. By the Levi's Lemma for
  traces~\cite[{Theorem 1}]{mazurkiewicz1995introduction}, we have that
  $t_1 \tequiv z_1z_2$, $t_2 \tequiv z_3z_4$, $s_1 \tequiv z_1z_3$,
  $s_2 \tequiv z_2z_4$ with $(z_2,z_3)\in I_D$. Since $(z_2,z_3)\in I_D$, we
  have $o \dtr{t_1}\dtr{z_3} o''$ and $o''\Uparrow$.
  By Lemma~\ref{aux-lem:commutative-observer-behavior},
  $o \dtr{s_1}\dtr{z_2} o''$, which contradicts the hypothesis that
  $o\Downarrow s$.
\end{proof}

The alternative characterisation for the \distributed\ preorder follows along
the lines of the one for the classical must testing preorder, when
Mazurkiewicz traces are considered instead of strings of actions.
For this reason, we extend the notions of transition relation, convergence and
residuals to Mazurkiewicz traces.

We now focus on the definition of the transition relation, which is
instrumental for the next definitions. We first note that, differently
from centralised must testing preorder, string inclusion may
not hold for the \distributed\ preorders. For
instance, take $p = \tau.a.\zero + \tau.b.\zero$ and $q = a.b.\zero$
over the interface $\mathbb{I}=\{\{a\}, \{b\}\}$.
Note that $p \dmustleq{\mathbb{I}} q$ because any \distributed\ observer that
passes $p$ is happy regardless of whether $a$ and $b$ are executed. Moreover,
such observer would be unable to detect if both $a$ and $b$ are executed
because those actions take place over different parts of the interface. Hence,
$\trc{q}\not\subseteq\trc{p}$ because $a.b\in\trc q$ but $a.b\not\in\trc p$
despite $p \dmustleq{\mathbb{I}} q$.
Our notion of reduction w.r.t a Mazurkiewicz trace will account for such
mismatch and, e.g., $p \dtr{\tclass{a.b}} \zero$ will hold.
Intuitively, the relation $\dtr {\tclass s}$ accounts for reductions in which
$s$ may be partially executed over some of the parts of the interface (but
complete in at least one of them).

We write $s <_D t$ iff $t\in \tclass{ss'}$ for some $s'\neq \epsilon$, i.e.,
$t$ extends $s$ (up-to the swapping of independent actions), and $s\leq_D t$
stands for $s <_D t$ or $s \equiv_D t$.
Then, the set of maximal reductions of $p$ within a trace $\tclass s$ is
\[
  \pref p {\tclass s}
  = \max_{<_D}\{t \ |\ t\in \trc p, t\leq_D s\}
\]
where $\max_{<_D}$ denotes the maximal elements of a set according to the
order ${<_D}$.
If $t\in \pref p {\tclass s}$ then $t$ cannot be extended with any symbol $a$
such that $t'\equiv_D ta\in \pref p {\tclass s}$ and $p\dtr {t'}$.

The following last ingredient allows us to ensure that different maximal
prefixes in $\pref p {\tclass s}$ actually execute the actions in $\tclass s$
(although some prefixes can be partial).
We recall that $\cohide$ stands for the operation that projects a string over
an alphabet.
Let $\mathbb{I}$ be the interface that induces the dependency relation $D$, we
write $\tsplit {\pref p {\tclass s}} {}$ if $\pref p {\tclass s}$
\emph{jointly-completes} ${\tclass s}$, which is defined by
\[
  \tsplit {\pref p {\tclass s}} {} \quad
  \iff \quad
  \forall I \in\mathbb{I}.\exists t\in \pref p {\tclass s}. s\cohide I = t\cohide
\]

\begin{exa}
  Consider the processes $p = \tau.a.\zero + \tau.b.\zero$ and $q = a.b.\zero$
  and the interface $\mathbb{I}=\{\{a\}, \{b\}\}$. Let $D$ be the dependency
  relation induced by $\mathbb{I}$. Then, we have
  \[
    \begin{array}{l@{\ = \ }l}
      \trc p & \{\epsilon,a,b\}\\
      \trc q & \{\epsilon,a,b,a.b\}
    \end{array}
  \]

  The restriction of $<_D$ to the elements in $\trc q$ is
  \[
    \{ (\epsilon, a), (\epsilon,b), (a, a.b), (b, a.b) \}
  \]

  Note that $b <_D ab$ because $b$ can be extended with $a$ and
  $ab \in \tclass{ba}$. Then,
  \[
    \begin{array}{l@{\ = \ }l}
      \pref {p} {\tclass {a.b}}
      & \max_{<_D}\{t \ |\ t\in\trc p, t \leq_D a.b\} =
        \max_{<_D}\{\epsilon,a,b\} = \{a,b\}\\
      \pref {q} {\tclass {a.b}}
      & \max_{<_D}\{t \ |\ t\in\trc q, t \leq_D a.b\} =
        \max_{<_D}\{\epsilon,a,b,a.b\} = \{a.b\}
    \end{array}
  \]
  Both cases jointly-completes $\tclass {a.b}$, i.e.,
  $\tsplit{\pref {p} {\tclass {a.b}}}{}$ and
  $\tsplit{\pref {q} {\tclass {a.b}}}{}$ do hold.

  On the contrary, if we consider $r = a.\zero$ we have that
  $\pref {r} {\tclass {a.b}} = \{a\}$, which does not jointly-complete
  $\tclass {a.b}$, because none of the strings in $\pref {r} {\tclass {a.b}}$
  matches $b$. \qed%
\end{exa}

Let $D$ be the dependency relation induced by the interface $\mathbb{I}$. We
let
\begin{itemize}
\item $p\dtr{\tclass s} p'$ if and only if
  $\tsplit{\mpref p {\tclass s}}{\tclass s}$ and
  $\exists s'\in\mpref p {\tclass s}$ {such that} $p \dtr {s'} {p'}$.
\item
  $p\Downarrow \tclass{s} \mbox{ if } \forall s'\in \tclass{s} \mbox{ then } p
  \Downarrow s'$.
\item $(p \after {\tclass s}) = \{p'\ |\ p\dtr{{\tclass s}} p'\}$.
\end{itemize}

\noindent
We adopt usual conventions for abbreviating notation when dealing with
transition relations and write, e.g., $p\dtr{\tclass s}$ instead of
\emph{there exists $p'$ such that} $p\dtr{\tclass s} p'$. Analogously,
$p \not\dtr{\tclass s}$ stands for \emph{there does not exist $p'$ such that}
$p\dtr{\tclass s} p'$.

In the definition below, the condition $L\subseteq I$ with $I\in\mathbb{I}$
captures the idea that each observation is relative to a specific part of the
interface.

\begin{defi}\label{def-charact-dmust} Let $\mathbb{I}$ be an interface and
  $D$ the dependency relation induced by $\mathbb{I}$. Then, $p\btmustleq q$
  if for every $s\in\actset^*$, for any part $I\in\mathbb{I}$, for all finite
  $L\subseteq I$, if $p\Downarrow \tclass s$ then
  \begin{enumerate}
  \item $q\Downarrow \tclass s$.
  \item $(p \after {\tclass s})\ \MUST\ L$ implies
    $(q \after {\tclass s})\ \MUST\ L$.
  \end{enumerate}
\end{defi}

\noindent
The following three lemmata are instrumental to the proof of the
correspondence theorem and characterise the relation between the Mazurkiewicz
traces of related processes.

\begin{lem}\label{lem-aux:conv-traces-swap} Let $\mathbb{I}$ be an interface
  and $D$ the dependency relation induced by $\mathbb{I}$. If
  $p\dmustleq {\mathbb I} q$ then for all $s\in\actset^*$ we have that
  $p\Downarrow\tclass{ s}$ implies
  \begin{enumerate}
  \item
    $q\Downarrow\tclass{ s}$.
  \item\label{lem-aux:trace-implication} $s\in\trc q$ implies that
    $p \dtr{\tclass s}$.
  \end{enumerate}
\end{lem}

\begin{proof} Assume $\mathbb{I} = {\{I_i\}}_{i\in 0 \ldots n}$ and $D$ the
  induced dependency relation.
  \begin{enumerate}
  \item By contradiction. Suppose there exists $s=a_1\ldots a_m$ such that
    $p\Downarrow \tclass s$ and $q\Uparrow \tclass s$. Then, take the observer
    $o = \Pi_{i \in 0, \ldots n} o_i$ with $o_i$ defined as follows
    \[ o_i = \tau.\one + \co{b^i_1}.(\tau.\one + \ldots (\tau.\one +
      \co{b^i_{k_i}}.\tau.\one)\ldots) \ \mbox{with}\ s\cohide {I_i} =
      b^i_1\ldots b^i_{k_i}
    \]

    For each maximal computation of $p \ \|\ o$, we proceed by unzipping the
    computation to conclude that $o \dtr{\co{t}}$ and $p\dtr{t}$ for some $t$.
    Note that $o\Downarrow t$ for all $t\in\actset^*$ and $o\dtr{\co{t}}$
    implies that there exists $t'$ such that $\co{t}t'\in\tclass{\co s}$.
    Since $p\Downarrow\tclass{ s}$, we have $p\Downarrow{t}$. Consequently,
    every maximal computation of $p \ \|\ o$ is finite, i.e.,
    $p \ \|\ o \dtr{}  p' \ \|\ o' \not\tr{}$. By induction on the length of
    $t$, it follows that $o \dtr{t} o'\not\tr{}$ implies
    $o' = \Pi_{i \in 0, \ldots n} \one$. Consequently, every maximal
    computation of $p \ \|\ o$ is observer-successful and $p\ \must\ o$.

    Since $q\Uparrow \tclass s$ there exists $tt'\in\tclass s$ such that
    $q\dtr{t}q_0$ and $q_0\Uparrow$. By induction on the length of $\co{t}$ it
    follows that there exists $o'$ such that $o\dtr{\co{t}}o'$ and
    $o' = \Pi_{i \in 0, \ldots n} o_i''$ where $o_i''=\tau.\one$ or
    $o_i''= (\tau.\one + a.(\ldots))$ for all $i$. Then, there exists a
    maximal (divergent) observer-unsuccessful computation
    $q \ \|\ o \tr{} q_0 \ \|\ o' \tr{} q_1 \ \|\ o' \tr{} \ldots \tr{} q_j \
    \|\ o' \tr{}\ldots$. Consequently, $q\notmust\ o$, which contradicts the
    assumption $p\dmustleq {\mathbb I} q$.

  \item By contradiction. Suppose that there exists $s=a_1\ldots a_m$ such
    that $p\Downarrow \tclass s$, $s\in\trc q$ and $p \not\dtr{\tclass s}$.
    From $p\not\dtr{\tclass s}$, either (i)
    $\tsplit{(\mpref p {\tclass s})}{\tclass s}$ does not hold; or (ii)
    $\not\exists s'\in{\mpref p {\tclass s}}$ {such that} $p \dtr {s'} {}$.
    Note that (ii) is impossible: because $\epsilon \in\trc{p}$ for all $p$
    and $\epsilon \leq_D s$ for all $s$ and $D$. Consequently, for all $p, s$
    and $D$ we have $\{t \ |\ t\in \trc p, t\leq_D s\}\neq \emptyset$ and
    $\mpref p {\tclass s} \neq \emptyset$. By the definition of
    $\mpref p {\tclass s}$, $s'\in\mpref p {\tclass s}$ implies
    $s'\in\trc{p}$, and therefore $p \dtr {s'} {}\!\!$, which contradicts
    (ii). Then, (i) holds. Consequently,
    \begin{equation}
      \exists I_j\in\mathbb{I}.\forall t\in(\mpref p {\tclass
        s}).s\cohide I_j \neq t\cohide I_j
    \end{equation}

    Then, choose the observer $o = \Pi_{i \in 0, \ldots n} o_i$ with $o_i$
    defined as follows
    \[
      \begin{array}{l@{\ =\ }l@{\qquad}l}
        o_i
        &
          \tau.\one
          + \co{b^i_1}.(\tau.\one +  \ldots (\tau.\one + \co{b^i_{k_i}}.\one)\ldots)
        & \mbox{ if } b^i_{k_i} \neq a_m
        \\
        o_i
        &
          \tau.\one
          + \co{b^i_1}.(\tau.\one +  \ldots (\tau.\one + \co{b^i_{k_i}}.\zero)\ldots)
        & \mbox{ if } b^i_{k_i} = a_m
      \end{array}
    \]
    with $s\cohide {I_i} = b^i_1\ldots b^i_{k_i}$.

    Since $s\in\trc q$, $q\dtr{s}q'$. By (1) above, $p\Downarrow \tclass s$
    implies $q\Downarrow{\tclass s}$. Consequently, $q'\Downarrow$. Then,
    there is a computation $q\dtr{s}q'\dtr{} q'' \not\tr{}$. Moreover, we can
    build an unsuccessful computation of
    $o\dtr{\co s} o'' = \Pi_{i \in 0, \ldots n} o_i''\not\tr{}$ where
    $o_j'' = \zero$ for $j\in 0,\ldots,n$ and $b^j_{k_j} = a_n$. By zipping
    the computations $q\dtr{s}q''$ and $o\dtr{\co s} o''$, we obtain a maximal
    computation of $q\ \|\ o$ that is observer-unsuccessful. Consequently,
    $q\ \notmust\ o$.

    Take a maximal computation of $p\ \|\ o \dtr{}$. By unzipping it,
    $p\dtr{t} p'$ and $o\dtr{\co{t} }o'$. By construction of $o$, $o\dtr{s}$.
    By Lemma~\ref{aux-lem:commutative-observer-behavior}, $o\dtr{r}$ iff
    $r\in\tclass s$. Then, $o\dtr{\co t}o'$ implies that there exists $t'$
    such that $\co{t}t'\in\tclass {\co s}$. From $p\Downarrow \tclass s$, we
    have $p'\Downarrow$ and, consequently, $p'\not\tr{}$. Also
    $o\Downarrow {\co t}$ holds because $o$ is finite; moreover,
    $o'\not\tr{}$. Hence, $p\ \|\ o \dtr{} p'\ \|\ o'$ is a finite maximal
    computation.

    By assumption, $\tclass s \cap \trc p =\emptyset$ holds (otherwise,
    $\tsplit{(\mpref p {\tclass s})}{\tclass s}$ should hold). Therefore
    $t\not\in\tclass{s}$. Then, $t$ is a prefix of a string in $\tclass s$,
    i.e., there exists $a$ and $t''$ such that $tt''a \in\tclass{s}$ and for
    all $p'$ if $p\dtr{t} p'$ then $p'\not\dtr{\tclass{t''a}}$. Without lost
    of generality we assume $a = a_n$ (otherwise, the definition of $o$ can be
    changed accordingly). By construction, $o \dtr{\co{t}} o'$ implies
    $o' = \Pi_{i \in 0, \ldots n} o'_i$ and $o'\dtr{\tclass{t''a}}$ where for
    any $i$ either (a) $o_i' = \tau.1 + \co{b}.(\ldots)$ or (b) $o_i' = \one$.
    Case (a) is not possible, because we assume $o'\not\tr{}$. For case (b),
    we proceed by zipping the computations $p\dtr{t'}p''\not\tr{}$ and
    $o \dtr{\co{t'}} o'\not\tr{}$, which is observer-successful.
    Consequently, every maximal computation of $p\ \|\ o$ is
    observer-successful and $p\ \must\ o$, which is in contradiction with the
    assumption $p\dmustleq {\mathbb I} q$.
    \qedhere
  \end{enumerate}
\end{proof}

\begin{lem}\label{lem-aux:must-empty}
  If $(p \after {\tclass s})\ \NOTMUST\ L$ for some finite
  $L\subseteq\actset$, then $p\dtr{\tclass s}$.
\end{lem}

\begin{proof}
  Suppose $p\not\dtr{\tclass s}$. Then $(p \after {\tclass s}) =\emptyset$
  and, by definition, $\emptyset\ \MUST\ L$ for every finite
  $L\subseteq \actset$.
\end{proof}

\begin{lem}\label{lem-aux:trace-implication-rec}
  If $p \btmustleq q$, $s\in\trc q$ and $p\Downarrow\tclass s$ then
  $p\dtr{\tclass s}$.
\end{lem}

\begin{proof}
  Assume that $p\not\dtr{\tclass s}$. Then,
  $(p\after {\tclass s}) = \emptyset$. Hence,
  $(p \after {\tclass s})\ \MUST\ L$ for every finite $L\subseteq A$. Since
  $p\Downarrow \tclass{s}$ and $p \btmustleq q$, $q\Downarrow \tclass{s}$. By
  straightforward induction on the length of the reduction, we can show that
  $q \dtr{t} q'$ implies $\n(q') \subseteq \n(q)$ for all $t$ (i.e., the names
  $\n(q')$ of $q'$ are included in the names of $q$). Moreover,
  $\init q \subseteq \n(q)$ trivially holds. Consequently,
  $\bigcup\{\init {q'}\ |\ q\dtr{\tclass s} q'\} \subseteq \n(q)$. Since,
  $\n(q)$ is finite, we can conclude that the set
  $\bigcup\{\init {q'}\ |\ q\dtr{\tclass s} q'\}$ is finite. Therefore, we can
  find an action $a$ such that $q\not\dtr{s a}$. Then
  $(q \after {\tclass s})\ \NOTMUST\ \{a\}$ while
  $(p \after {\tclass s})\ \MUST\ \{a\}$, which contradicts the hypothesis
  $p \btmustleq q$.
\end{proof}

\medskip
\begin{thm}\label{th:comm-characterisation}
  ${\dmustleq {\mathbb{I}}} = {\btmustleq}$.
\end{thm}

\begin{proof}
  \hfill
  \begin{itemize}[align=left]
  \item[($\subseteq$)] Actually we prove that $p\not\btmustleq q$ implies
    $p \not\dmustleq {\mathbb{I}} q$. Let $D$ be the dependency relation
    induced by $\mathbb{I}$. Assume that there exists $s=a_1\ldots a_n$ and
    $I_j\in\mathbb{I}$ and $L\subseteq I_j$ such that
    \begin{enumerate}
    \item
      $p\Downarrow \tclass s$ and $q\Uparrow \tclass s$, or
    \item
      $s \in \trc q$ and $\forall t\in\tclass s. t\not\in\trc p$ or
    \item
      $(p \after \tclass s) \  \MUST\ L$ and $(q \after \tclass s) \NOTMUST\ L$
    \end{enumerate}
    For each case we show that there exists an observer such that
    $p\ \must\ o$ and $q \notmust\ o$. For the two first cases, we take the
    observers as defined in proof of Lemma~\ref{lem-aux:conv-traces-swap}. For
    the third one, we take $o = \Pi_{i \in 0, \ldots n} o_i$ with $o_i$
    defined as follows
    \[
      \begin{array}{l@{\ =\ }l@{\qquad}l}
        o_i
        &
          \tau.\one + b_1.(\tau.\one +  \ldots (\tau.\one + b_k.\one )\ldots)
        & \mbox{ if } i \neq j
        \\
        o_i
        &
          \tau.\one + b_1.(\tau.\one
          +  \ldots (\tau.\one + b_k.\sum_{a\in L} a.\one)\ldots)
        &
          \mbox{ if } i = j
      \end{array}
    \]

    with $s\cohide {I_i} = b_1\ldots b_k$.
  \item[($\supseteq$)] We prove $p \btmustleq q$ implies
    $p \dmustleq {\mathbb{I}} q$. Actually, the proof follows by showing that
    $p \btmustleq q$ and $q\ \notmust\ o$ imply $p\ \notmust\ o$. Assume there
    exists an unsuccessful computation
    \[
      q\ \|\ o = q_0\ \|\ o_0\tr{\tau}\ldots \tr{\tau}q_k\
      \|\ o_k\tr{\tau}\ldots
    \]
    Consider the following cases:

    \begin{enumerate}
    \item\label{proof:th:comm-characterisation-item1} {\bf The computation is
        finite}, i.e., there exists $n$ such that $q_n\ \|\ o_n\not\tr{\tau}$.

      By unzipping the computation we have $q_0\dtr{s} q_n$ and
      $o_0\dtr{\co{s}}o_n$, which is unsuccessful, i.e.,
      $o_i\not\tr{\success}$ for all $0 \leq i \leq n$.

      Moreover, $q_n \NOTMUST\ \init {o_n}$. Hence
      $(q\ \after \ \tclass s) \NOTMUST\ \init{o_n}$.
      \begin{enumerate}
      \item Case $p\Uparrow \tclass s$, i.e., $\exists t\in\tclass s$ and
        $p\Uparrow t$. By
      Corollary~\ref{aux-cor:commutative-observer-behavior}~(\ref{aux-lem:comm-observer-unsuccessful-comp}),
        $o\dtr {\co{s}}$ implies $o\dtr{\co t}$ also unsuccessful, and hence
        there is an unsuccessful computation of $p\ \|\ o$.

      \item Case $p\Downarrow \tclass s$. Note that $s\in\trc q$. By
        Lemma~\ref{lem-aux:conv-traces-swap}~(\ref{lem-aux:trace-implication}),
        $\exists t\in\tclass s:t\in\trc p$. Hence,
        $(p\ \after \ \tclass s)\neq\emptyset$. Since $p \btmustleq q$,
        $(q\ \after \ \tclass s)\NOTMUST\ \init{o_n}$ implies
        $(p\ \after \ \tclass s) \NOTMUST\ \init{o_n}$. Therefore, exists some
        $p'\in (p\ \after \ \tclass s)$ and $p' \NOTMUST\ \init {o_n}$. Hence,
        $\exists t'\in\tclass s. p\dtr{t'}$. By
        Corollary~\ref{aux-cor:commutative-observer-behavior}~(\ref{aux-lem:comm-observer-unsuccessful-comp}),
        $o\dtr{\co t'}$ unsuccessful, and hence there is an unsuccessful
        computation of $p\ \|\ o$.
      \end{enumerate}

    \item {\bf The computation is infinite}. We consider two cases:

      \begin{enumerate}
      \item There exist $s\in\trc{ q}$ and $\co{s}\in\trc{o}$ such that
        $q\Uparrow s$ or $o\Uparrow {\co s}$. We proceed by case analysis.
	\begin{itemize}
	\item[(i)] $q\Uparrow \tclass s$: Since $p \btmustleq q$,
          $p\Uparrow \tclass s$. Therefore, $\exists t\in\tclass s$ such that
          $p\Uparrow t$. By
          Corollary~\ref{aux-cor:commutative-observer-behavior}~(\ref{aux-lem:comm-observer-unsuccessful-comp}),
          $o\dtr{\co t}$ unsuccessful, and hence there is an unsuccessful
          computation of $p\ \|\ o$.
	\item[(ii)] $q\Downarrow \tclass{s}$ (and $o\Uparrow {\co s}$): By
          Lemma~\ref{lem-aux:trace-implication-rec},
          $\exists t\in\tclass s: t\in\trc p$. By
          Corollary~\ref{aux-cor:commutative-observer-behavior}~(\ref{aux-lem:comm-observer-preserve-divergence}),
          $o\Uparrow{\co t}$, and hence there is an unsuccessful computation
          of $p\ \|\ o$.
	\end{itemize}

      \item $\forall n. q_n\Downarrow$ and $o_n\Downarrow$. For every $n$,
        take $s\in\actset^*$ such that $q\dtr{s}q_n$ and $q\Downarrow s$ (this
        is possible because $q \|\ o \dtr{} q_n \|\ o_n$ is unsuccessful and
        $\forall i\leq n. q_i\Downarrow)$. By
        Lemma~\ref{lem-aux:conv-traces-swap}, $q\dtr{s}q_n$ and
        $q\Downarrow s$ implies either (i) $p\Uparrow \tclass s$ or (ii)
        $p\Downarrow \tclass s$ and $p\dtr{\tclass s}$.
        \begin{itemize}
  	\item[(i)] $p \Uparrow \tclass s$: $\exists t\in\tclass s$ such that
          $p\Uparrow t$. By
          Corollary~\ref{aux-cor:commutative-observer-behavior}~(\ref{aux-lem:comm-observer-unsuccessful-comp}),
          $o\dtr{\co s}$ unsuccessful implies $o\dtr{\co t}$ unsuccessful, and
          hence there is an unsuccessful computation of $p\ \|\ o$.
	\item[(ii)] $p\Downarrow \tclass s$ and $p\dtr{\tclass s}$. Since
          $q_n\Downarrow$, there exists $q_m$ such that
          $q_n \dtr{}q_m \not\tr{}$ and $q_m\tr{a}q_{m+1}$. Consequently,
          $(q\ \after \ \tclass s) \NOTMUST\ L$ for all $L$ such that
          $a\not\in L$. Since $p \btmustleq q$, then for all $L$ such that
          $a\not\in L$, we have $(p\ \after \ \tclass s) \NOTMUST\ L$.
          Moreover, $o_n \dtr{} o_m \tr{\co{a}} o_{m+1}$. Therefore, there
          exists $p_n \in(p\ \after \ \tclass s) $ and $p_n\dtr{a} p_{m+1}$.
          Hence, for all $n$ there is an unsuccessful computation
          $p \|\ o \dtr{} p_n\|\ o_n \dtr{} p_{m+1} \|\ o_{m+1}$. \qedhere
        \end{itemize}
      \end{enumerate}
    \end{enumerate}
  \end{itemize}
\end{proof}

\noindent
In the following we will write $L_{p, \tclass s}^I$ for the smallest set such
that $(p \after \tclass s)\ \MUST \ L$ and $L\subseteq I$ imply
$L_{p,\tclass s}^I\subseteq L$.

\begin{exa}\label{ex:trip-dmust} We take advantage of the alternative
  characterisation of the \distributed\ preorder to show that the three
  processes for the broker in Example~\ref{ex:trip} are equivalent when
  considering
  $\mathbb{I}=\{\{\mathit{req}\}, \{\mathit{reqF}\}, \{\mathit{reqH}\}\}$.
  Actually, we will only consider $B_0 \dmusteq{\mathbb{I}} B_1$, as the
  proofs for $B_0\dmusteq{\mathbb{I}} B_2$ and $B_1 \dmusteq{\mathbb{I}} B_2$
  are analogous.

  Firstly, we have to consider that $B_0\Downarrow s$ and $B_1\Downarrow s$
  for any $s$ because $B_0$ and $B_1$ do not have infinite computations. The
  relation between must-sets are described in the two tables below. The first
  table shows the sets $(B_0 \after \tclass s)$ and $L_{B_0,\tclass s}^I$.
  Note that $\tclass s$ in the first column will be represented by any string
  $s'\in\tclass s$. Moreover, we write ``$-$'' in the three last columns
  whenever $L_{B_0, \tclass s}^I$ does not exist. The second table does the
  same for $B_1$. In the tables, we let $B_0'$ stand for
  $\tau.\co{\mathit{reqF}}.\zero+ \tau.\co{\mathit{reqH}}.\zero +
  \tau.\co{\mathit{ reqH}}.\co{\mathit{reqF}}.\zero$ and $B_1'$ stand for
  $\tau.\co{\mathit{reqF}}.\zero+ \tau.\co{\mathit{reqH}}.\zero+
  \tau.\co{\mathit{reqF}}.\co{\mathit{reqH}}.\zero$.

  \begin{center}
    \begin{math}
      \begin{array}{lcccc}
        \toprule
        \tclass s
        & B_0\after \tclass s
        & L_{B_0,\tclass s}^{\{req\}}
        & L_{B_0,\tclass s}^{\{reqH\}}
        & L_{B_0, \tclass s}^{\{reqF\}}
        \\
        \midrule
        \epsilon
        & B_0
        & \{\mathit{req}\}
        & -
        & -
        \\
        \mathit{req}
        &
          \{B_0', \co{\mathit{reqF}}.\zero, \co{\mathit{reqH}}.\zero,
          \co{\mathit{reqH}}.\co{\mathit{reqF}}.\zero\}
        & -
        & -
        & -
        \\
        \mathit{req}.\co{\mathit{reqF}}
        & \{0\}
        & -
        & -
        & -
        \\
        \mathit{req}.\co{\mathit{reqH}}
        &
        \{0, \co{\mathit{reqF}}.\zero \}
        & -
        & -
        & -
        \\
        \mathit{req}.\co{\mathit{reqF}}.\co{\mathit{reqH}}
        & \{0\}
        & -
        & -
        & -
        \\
        {\sf other}
        & \emptyset
        & \emptyset
        & \emptyset
        & \emptyset
        \\
        \bottomrule
      \end{array}
    \end{math}
  \end{center}

  \begin{center}
    \begin{math}
      \begin{array}{lcccc}
        \toprule
        \tclass s
        & B_1\after \tclass s
        & L_{B_0,\tclass s}^{\{req\}}
        & L_{B_0,\tclass s}^{\{reqH\}}
        & L_{B_0,\tclass s}^{\{reqF\}}
        \\
        \midrule
        \epsilon
        & B_1
        & \{\mathit{req}\}
        & -
        & -
        \\
        \mathit{req}
        & \{B_1', \co{\mathit{reqF}}.\zero, \co{\mathit{reqH}}.\zero,
          \co{\mathit{reqF}}.\co{\mathit{reqH}}.\zero\}
        & -
        & -
        & -
        \\
        \mathit{req}.\co{\mathit{reqF}}
        & \{0, \co{\mathit{reqH}} \}
        & -
        & -
        & -
        \\
        \mathit{req}.\co{\mathit{reqH}}
        & \{ 0\}
        & -
        & -
        & -
        \\
        \mathit{req}.\co{\mathit{reqF}}.\co{\mathit{reqH}}
        & \{0\}
        & -
        & -
        & -
        \\
        {\sf other}
        & \emptyset
        & \emptyset
        & \emptyset
        & \emptyset
        \\
        \bottomrule
      \end{array}
    \end{math}
  \end{center}

  By inspecting the tables, we can check that for any possible trace
  $\tclass s$ and $I\in\mathbb{I}$, it holds that
  $L_{B_0, \tclass s}^I = L_{B_1, \tclass s}^I$. Consequently,
  $(B_0 \after \tclass s)\ \MUST\ L$ iff $(B_1 \after \tclass s)\ \MUST\ L$
  and thus we have $B_0 \dmusteq{\mathbb{I}} B_1$.
  \qed%
\end{exa}

We now present two additional examples that help us understand the
discriminating capability of the \distributed\ preorder and its differences
with the classical must preorder.

The first of these examples shows that a process that does not communicate its
internal choices over all parts of its interface is useless in a distributed
context.

\begin{exa}
  Consider the process $p = \tau.a.\zero + \tau.b.\zero$ that is intended to
  be used by two processes with the following interface:
  $\mathbb{I} = \{\{a\},\{b\}\}$. We show that this process is less useful
  than $0$ in an \distributed\ context, i.e.,
  $ \tau.a.\zero + \tau.b.\zero \dmustleq{\mathbb{I}} 0$. It is immediate to
  see that $p$ and $0$ strongly converge for any $s\in\actset^*$, then the
  minimal sets $L_{p, \tclass s}^{\{a\}}$, $L_{p,\tclass s}^{\{b\}}$,
  $L_{0, \tclass s}^{\{a\}}$ and $L_{0,\tclass s}^{\{b\}}$ presented in the
  tables below are sufficient for proving our claim.

  \begin{center}
    \begin{math}
      \begin{array}{lcccc}
        \toprule
        \tclass s
        & p \after \tclass s
        & L_{p, \tclass s}^{\{a\}}
        & L_{p, \tclass s}^{\{b\}}
        \\
        \midrule
        \epsilon
        & p,a,b
        & -
        & -
        \\
        a
        & \{0\}
        & -
        & -
        \\
        b
        & \{0\}
        & -
        & -
        \\
        {\sf other}
        & \emptyset
        & \emptyset
        & \emptyset
        \\
        \bottomrule
      \end{array}
      \qquad
      \begin{array}{lcccc}
        \toprule
        \tclass s
        & 0 \after \tclass s
        & L_{0, \tclass s}^{\{a\}}
        & L_{0, \tclass s}^{\{b\}}
        \\
        \midrule
        \epsilon
        & 0
        & -
        & -
        \\
        a
        & \emptyset
        & \emptyset
        & \emptyset
        \\
        b
        & \emptyset
        & \emptyset
        & \emptyset
        \\
        {\sf other}
        & \emptyset
        & \emptyset
        & \emptyset
        \\
        \bottomrule
      \end{array}
    \end{math}
  \end{center}

  Note that differently from the classical must preorder, the \distributed\
  preorder does not consider the must-set $\{a,b\}$ to distinguish $p$ from
  $0$ because this set involves channels in different parts of the interface.
  The key point here is that each internal reduction of $p$ is observed just
  by one part of the interface: the choice of branch $a$ is only observed by
  one process and the choice of $b$ is observed by the other one. Since
  \distributedobs\ observers do not intercommunicate, they can only report
  success simultaneously if they can do it independently from the interactions
  with the tested process, but such observers are exactly the ones that $0$
  can pass.

  Like in the classical must preorder, we have that
  $0 \not\dmustleq{\mathbb{I}} \tau.a.\zero + \tau.b.\zero$. This is witnessed
  by the observer $o= \co{a}.\zero + \tau.\one\ \|\ \one$ that is passed by
  $0$ but not by $\tau.a.\zero + \tau.b.\zero$. \qed%
\end{exa}

The second example shows that the \distributed\ preorder falls somehow short
with respect to the target we set in the introduction of allowing processes to
swap actions that are targeted to different partners.

\begin{exa}\label{ex-dmust-swap-axiom}
  Consider the interface $\mathbb{I} = \{\{a\},\{b\}\}$ and the two pairs of
  processes
  \begin{itemize}
  \item
    $a.b.\zero + a.\zero+ b.\zero$ and
    $b.a.\zero + a.\zero + b.\zero$.
  \item
    $a.b.\zero$ and $b.a.\zero$.
  \end{itemize}

  \noindent
  By inspecting traces and must-sets in the two tables below, where we use $p$
  and $q$ to denote $ a.b.\zero + a.\zero + b.\zero$ and
  $b.a.\zero + a.\zero + b.\zero$

  \begin{center}
    \begin{math}
      \begin{array}{ll}
        \begin{array}{lcccc}
          \toprule
          \tclass s
          & p\after \tclass s
          & L_{p,\tclass s}^{\{a\}}
          & L_{p,\tclass s}^{\{b\}}
          \\
          \midrule
          \epsilon
          & \{p\}
          & \{a\}
          & \{b\}
          \\
          a
          & \{b.\zero, 0\}
          & -
          & -
          \\
          b
          & \{0\}
          & -
          & -
          \\
          ab
          & \{0\}
          & -
          & -
          \\
          {\sf other}
          & \emptyset
          & \emptyset
          & \emptyset
          \\
          \bottomrule
        \end{array}
        &\qquad
          \begin{array}{lcccc}
            \toprule
            \tclass s
            & q\after \tclass s
            & L_{q,\tclass s}^{\{a\}}
            & L_{q,\tclass s}^{\{b\}} \\
            \midrule
            \epsilon
            & \{p\}
            & \{a\}
            & \{b\}
            \\
            a
            & \{0\}
            & -
            & -
            \\
            b
            & \{a.\zero, 0\}
            & -
            & -
            \\
            ab
            & \{0\}
            & -
            & -
            \\
            {\sf other}
            & \emptyset
            & \emptyset
            & \emptyset
            \\
            \bottomrule
          \end{array}
      \end{array}
    \end{math}
  \end{center}

  it is easy to see that
  \[
    a.b.\zero + a.\zero + b.\zero
    \dmusteq {\mathbb{I}}
    b.a.\zero +  a.\zero + b.\zero
  \]
  However, by using $o= \co{a}.\one\ \|\ \one$ and $o'= \one\ \|\ \co{b}.\one$
  as observers, it can be shown that
  \[
    a.b.\zero \not\dmustleq{\mathbb{I}} b.a.\zero
    \qquad \mbox{and} \qquad
    b.a.\zero \not\dmustleq{\mathbb{I}} a.b.\zero
  \]

  Note that $o = \co{a}.\one\ \|\ \one$ actually interacts with the process
  under test by using just one part of the interface and relies on the fact
  that the remaining part of the interface stays idle. Thanks to this ability,
  \distributedtest\ observers have still a limited power to track some
  dependencies among actions on different parts of the interface.

  The preorder presented in the next section limits further the discriminating
  power of observers and allows us to equate processes $a.b.\zero$ and
  $b.a.\zero$. \qed%
\end{exa}


\section{A testing preorder with \uncoordinated\ observers}\label{sec:uncoordinated}
In this section we explore a notion of equivalence equating processes that can
freely permute actions over different parts of their interfaces. As for the
\distributed\ observers, the targeted scenario is that of a service with a
partitioned interface interacting with two or more independent processes by
using separate sets of ports. In addition, each component of an observer
cannot exploit any knowledge about the design choices made by the other
components, i.e., each of them has a local view of the behaviour of the
process that ignores all actions controlled by the remaining components. Local
views are characterised in terms of a projection operator defined as follows.

\begin{defi}[Projection]
  Let $V \subseteq \mathcal{N}$ be a set of observable ports. We write
  $ p \abst V$ for the process obtained by hiding all actions of $p$ over
  channels that are not in $V$. Formally,

  \begin{center}
    \begin{math}
      \begin{array}{l}
        \anmathrule{p\tr{\alpha} p' \quad \alpha\in V\cup\co V}
                   {p\abst V\tr{\alpha} p'\abst V}
                   \qquad\qquad
                   \anmathrule{p\tr{\alpha} p' \quad \alpha\not\in V\cup\co V}
                              {p\abst V\tr{\tau} p'\abst V}
      \end{array}
    \end{math}
  \end{center}
\end{defi}

\begin{exa}
  Let $p = \co a.p_1 + b.p_2$ be a process. Note that $p \tr{\co a} p_1$ and
  $p \tr{b} p_2$. Then, the projection of $p$ over the channel $a$, i.e.,
  $p\abst\{a\}$, has the following two transitions:
  $p\abst\{a\} \tr{\co a} p_1\abst\{a\}$ and
  $p\abst\{a\} \tr{\tau} p_1\abst\{a\}$ where the action $\co a$ of $p$ over
  the visible channel $a$ is reflected on the label of the transition of
  $p\abst\{a\}$ while the action over the non-visible channel $b$ is taken as
  an internal action.\qed%
\end{exa}

\begin{defi}[\Uncoordinated\ (must) preorder $\lmustleq{\mathbb{I}}$]
  Let $\mathbb{I} = {\{I_i\}}_{i\in 0..n}$ be an interface. We say
  $p\lmustleq{\mathbb{I}} q$ iff for all \distributedobs\ observer
  $o = \Pi_{i \in 0..n} o_i $, for all $i\in0..n$, $p\abst I_i \ \must\ o_i\ $
  implies $q\abst I_i\ \must \ o_i$.
\end{defi}

Note that $a.b.\zero$ and $b.a.\zero$ cannot be distinguished anymore by the
observer $o = \co{a}.\one\ \|\ \one$ used in the previous section to prove
$a.b.\zero \not\dmustleq{\{\{a\},\{b\}\}} b.a.\zero$
(Example~\ref{ex-dmust-swap-axiom}), because
$a.b.\zero\abst \{a\}\ \must\ \co{a}.\one$,
$b.a.\zero\abst \{a\}\ \must\ \co{a}.\one$,
$a.b.\zero\abst \{b\}\ \must\ \one$ and $b.a.\zero\abst \{b\}\ \must\ \one$.
Indeed, later (Example~\ref{ab-ind-equiv-ba}) we will see that:
\[a.b.\zero\lmusteq{{\{\{a\},\{b\}\}}}b.a.\zero.\]

\subsection{Semantic characterisation}

In this section we address the characterisation of the \uncoordinated\
preorder in terms of traces. We start by introducing an equivalence relation
over traces that ignores hidden actions.

\begin{defi}[\Filteredtr]
  Let $I\subseteq\actset$. Two strings $s, t\in\actset^*$ are \emph{equivalent
    up-to} $I$, written $s\noisyeq[I] t$, if and only if
  $s\cohide I = t\cohide I$. We write $\nclass[I] s$ for the equivalence class
  of $s$.
\end{defi}

Basically, two traces are equivalent up-to $I$ when they coincide after the
removal of hidden actions. For instance,
${aa} \ \noisyeq[\{a\}] {aba} \ \noisyeq[\{a\}] {ababbb} \noisyeq[\{a\}]
{\ldots}$.

As for the distributed preorder, we extend the notions of reduction,
convergence and residuals to equivalence classes of filtered traces.
\begin{itemize}
\item
  $q\dtr{\nclass s} q'\ \mbox{ if and only if } \ \exists t\in\nclass s\
  \mbox{ such that } \ q \dtr t {q'}$.
\item
  $p\Downarrow\nclass{s} \ \mbox{ if and only if }\ \forall t\in
  \nclass{s}.p \Downarrow t$.
\item $(p \after {\nclass s}) = \{p'\ |\ p\dtr{{\nclass s}} p'\}$.
\end{itemize}

\noindent
The following auxiliary result establishes properties relating reductions,
hiding and \filteredtr, which will be useful in the proof of the
correspondence theorem.
\begin{lem}\label{lem-prop-noisy}
  \hfill
  \begin{enumerate}
  \item\label{lem-red-noisy-class} $p\dtr{s}p'$ implies
    $p\abst I\ \dtr{s\cohide I}p'\abst I$.

  \item\label{lem-red-noisy-class-rec}
    $p\abst I\ \dtr{s}p'\abst I$ implies $\exists t\in\nclass s$ and
    $p\dtr{t}p'$.

  \item\label{lem-div-noisy-class} $p\Uparrow\nclass s$ implies
    $p\abst I \Uparrow s\cohide I$.

  \item\label{lem-after-noisy-class}
    $(p\ \after\ \nclass s)\ \MUST\ L\ \mbox{ iff }\ (p\abst I\after\ s\cohide
    I)\ \MUST\ L\cap I$.
  \end{enumerate}
\end{lem}

\begin{proof} The proof follows by induction on the length of $s$.
\end{proof}

The alternative characterisation for the \uncoordinated\ preorder is given in
terms of \filteredtr.

\begin{defi}
  Let $p\bnmustleq[\mathbb{I}] q$ if for every $I \in\mathbb{I}$, for every
  $s\in I^*$, and for all finite $L\subseteq I$, if $p\Downarrow\nclass s$
  then
  \begin{enumerate}
  \item
    $q\Downarrow\nclass s$
  \item\label{def-cond:futureU}q
    $(p \after {\nclass s})\ \MUST\ L\cup(\actset\backslash I)$ implies
    $(q \after {\nclass s})\ \MUST\ L\cup(\actset\backslash I)$
  \end{enumerate}
\end{defi}

\noindent
We would like to draw attention to condition~\ref{def-cond:futureU} above; it
only considers must-sets that always include all the actions in
$(\actset\backslash I)$ to avoid the possibility of distinguishing reachable
states because of actions that are not in $I$. Consider that this condition
could be formulated as follows: for all finite $L\subseteq \actset$,
\[
  (p \after {\nclass s})\ \MUST\ L
  \mbox{ implies }
  \exists L' \mbox{ such that }
  (q \after {\nclass s})\ \MUST\ L' \mbox{ and } L\cap I = L' \cap I
\]
which makes evident that only the actions from the observable part of the
interface are relevant. We adopted the first formulation because it resembles
the original characterisation of the must preorder.

The following lemmata are analogous to those for the \distributed\ preorder
and their proof follows similarly (proof details are in
Appendix~\ref{ap:local-proofs}).

\begin{restatable}{lem}{convnoisy}\label{lem-aux:conv-noisy}
  If $p\lmustleq{\mathbb{I}} q$ then for all $s\in\actset^*$ and
  $I\in\mathbb{I}$, we have that $p\Downarrow\nclass{ s}$ implies
  \begin{enumerate}
  \item
    $q\Downarrow\nclass{ s}$
  \item\label{lem-aux:trace-implication-noisy}
    $s\in\trc q$ implies that there exists $t\in \nclass s$ such that
    $t\in\trc p$.
  \end{enumerate}
\end{restatable}

\begin{restatable}{lem}{mustemptynoisy}\label{lem-aux:must-empty-noisy}
  if $(p \after {\nclass s})\ \NOTMUST\ L$ for some $L\subseteq \actset$, then
  $\exists t\in\nclass s: t\in\trc p$.
\end{restatable}

We rely on the following auxiliary results relating the traces of
processes in the must preorders.

\begin{restatable}{lem}{traceimplicationrecnoisy}\label{lem-aux:trace-implication-rec-noisy}
  If $p \bnmustleq[\mathbb{I}] q$, $s\in\trc q$ and $p\Downarrow\nclass s$
  with $I\in\mathbb{I}$ then $ t\in (\nclass s \cap \trc p)$.
\end{restatable}

\medskip
\begin{restatable}{thm}{thcharacterisationlocalmust}
  ${\lmustleq{\mathbb{I}}} = {\bnmustleq[\mathbb{I}]}$.
\end{restatable}

\begin{proof} The proof follows along the lines of that of
  Theorem~\ref{th:comm-characterisation} (see details in
  Appendix~\ref{ap:local-proofs}).
\end{proof}

\begin{exa}\label{ab-ind-equiv-ba} Consider the processes $p = a.b.\zero$ and
  $q = b.a.\zero$ and the interface $\mathbb{I} = \{\{a\},\{b\}\}$. The table
  below shows the analysis for the part of the interface $\{a\}$.
  \begin{center}
    \begin{math}
      \begin{array}{lcccc}
        \toprule
        \nclass[{\{a\}}] s
        & p \after \nclass[{\{a\}}] s
        & L_{p,\nclass s}^{\{a\}}
        & q \after \nclass[{\{a\}}] s
        &  L_{q,\nclass s}^{\{a\}}
        \\
        \midrule
        \epsilon
        & \{p\}
        & \{a\}
        & \{q,a.\zero\}
        & \{a\}
        \\
        a
        & \{0,b.\zero\}
        & -
        &\{0\}
        & -
        \\
        {\sf other}
        & \emptyset
        & \emptyset
        & \emptyset
        & \emptyset
        \\
        \bottomrule
      \end{array}
    \end{math}
  \end{center}

  When analysing the sets $(p \after {\nclass[{\{a\}}] \epsilon}) = \{p\}$ and
  $(q \after {\nclass[{\{a\}}] \epsilon}) = \{q, a.\zero\}$, we ignore the
  fact that $q$ starts with a hidden action $b$; the only relevant residuals
  are those performing $a$. With a similar analysis we conclude that the
  condition on must-sets also holds for set $\{b\}$. Hence,
  $ a.b.\zero \lmusteq{\mathbb{I}} b.a.\zero$ holds. \qed%
\end{exa}

The following example illustrates also the fact that \uncoordinated\ observers
are unable to track causal dependencies between choices made in different
parts of the interface.

\begin{exa}\label{ex:dist-ac+bd} Let $p_1 = a.c.\zero+b.d.\zero$ and
  $p_2 = a.d.\zero+b.c.\zero$ be two alternative implementations for a service
  with interface $\mathbb{I}= \{\{a,b\},\{c,d\}\}$. These two processes are
  distinguished by the \distributed\ preorder
  ($p_1\not\dmusteq{\{\{a,b\},\{c,d\}\}} p_2$) because of the observers
  $o_1= \co{a}.\one\ \|\ \co{c}.\one$
  ($p_1\not\dmustleq{\{\{a,b\},\{c,d\}\}} p_2$) and
  $o_2= \co{b}.\one\ \| \ \co{c}.\one$
  ($p_2\not\dmustleq{\{\{a,b\},\{c,d\}\}} p_1$).

  They are instead equated by the \uncoordinated\ preorder with respect to
  $\mathbb{I}$, $p_1\lmusteq{\mathbb{I}} p_2$. Indeed, if only the part of the
  interface $\{a,b\}$ is of interest, we have that $p_1$ and $p_2$ are
  equivalent because they exhibit the same interactions over channels $a$ and
  $b$. Similarly, without any a priori knowledge of the choices made for
  $\{a, b\}$, the behaviour observed over $\{c,d\}$ can be described by the
  non-deterministic choice $ \tau.{c}.\zero+\tau.{d}.\zero$, and hence, $p_1$
  and $p_2$ are indistinguishable also over $\{c,d\}$.

  We use the alternative characterisation to prove our claim. As usual,
  $p_1\Downarrow s$ and $p_2\Downarrow s$ for any $s$. The tables below show
  coincidence of the must-sets. We would only like to remark that
  $ac \in \nclass[\{a,b\}] a$ and, consequently,
  $p_1\after \nclass[\{a,b\}] a$ contains also process $0$.

  \begin{center}
    \begin{math}
      \begin{array}{lcccc}
        \toprule
        \nclass[\{a,b\}] s
        & p_1\after \nclass[\{a,b\}] s
        & L_{p_1,\nclass s}^{\{a,b\}}
        & p_2\after \nclass[\{a,b\}] s
        & L_{p_2,\nclass s}^{\{a,b\}}
        \\
        \midrule
        \epsilon
        & p_1
        & \{\mathit{a, b}\}
        & p_2
        & \{\mathit{a, b}\}
        \\
        a
        &\{c.\zero, 0\}
        & -
        & \{d.\zero, 0\}
        & -
        \\
        b
        & \{d.\zero, 0\}
        & -
        & \{c.\zero, 0\}
        & -
        \\
        {\sf other}
        & \emptyset
        & \emptyset
        & \emptyset
        & \emptyset
        \\
        \bottomrule
      \end{array}
    \end{math}
  \end{center}

  \begin{center}
    \begin{math}
      \begin{array}{lcccc}
        \toprule
        \nclass[\{c,d\}] s
        & p_1\after \nclass[\{a,b\}] s
        & L_{p_1,\nclass s}^{\{c,d\}}
        & p_2\after \nclass[\{a,b\}] s
        & L_{p_2,\nclass s}^{\{c,d\}}
        \\
        \midrule
        \epsilon
        & p_1
        & \{\mathit{c, d}\}
        & p_2
        & \{\mathit{c, d}\}
        \\
        c
        & \{0\}
        & -
        & \{0\}
        & -
        \\
        d
        & \{0\}
        & -
        & \{0\}
        & -
        \\
        {\sf other}
        & \emptyset
        & \emptyset
        & \emptyset
        & \emptyset
        \\
        \bottomrule
      \end{array}
    \end{math}
  \end{center}
  \qed%
\end{exa}


\section{Relation between must, \distributed\ and \uncoordinated\
  preorders}\label{sec:relation}
In this section, we formally study the relationships between the classical
must preorder and the two preorders we have introduced.
We start by showing that a refinement of an interface induces a coarser
preorder, e.g., splitting the observation among more \distributed\ observers
decreases the discriminating power of the observers. We say that an interface
${\mathbb{I}'}$ is a \emph{refinement} of another interface ${\mathbb{I}}$
when the partition ${\mathbb{I}'}$ is finer than the partition ${\mathbb{I}}$.

\begin{lem}\label{lemm-refinement-dist}
  Let $\mathbb{I}$ be an interface and $\mathbb{I}'$ a refinement of
  $\mathbb{I}$. Then, $p \dmustleq{\mathbb{I}} q$ implies
  $p \dmustleq{\mathbb{I}'} q$.
\end{lem}

\begin{proof}
  The proof follows by showing that $p\btmustleq[\mathbb{I}] q $ implies
  $p\btmustleq[\mathbb{I}'] q $. Let $D$ and $D'$ be the dependency relations
  induced respectively by $\mathbb{I}$ and $\mathbb{I}'$. Since $\mathbb{I}'$
  is a refinement of $\mathbb{I}$, $D'\subseteq D$ and therefore
  $\tclass[D] s \subseteq \tclass[D'] s$ for all $s$. Assume
  $p\Downarrow \tclass[D'] s$ for $s\in\actset^*$. Then,
  \begin{itemize}
  \item $p\Downarrow t$ for all $t\in\tclass[D'] s$. Note that
    $\tclass[D'] t = \tclass[D'] s$ because $t\in\tclass[D'] s$. Consequently,
    $p\Downarrow \tclass[D] t$ because
    $\tclass[D] t \subseteq \tclass[D'] t = \tclass[D'] s$. Since
    $p\btmustleq[\mathbb{I}] q $, we know that $q\Downarrow \tclass[D] t $,
    which implies $q \Downarrow t$. Therefore, $q\Downarrow \tclass[D'] s $.
  \item Assume $L\subseteq I'$, $I'\in\mathbb{I}'$ and
    $(p \after {\tclass[D'] s})\ \MUST\ L$. Then, $(p \after t)\ \MUST\ L$ for
    all $t\in \tclass [D'] s$. Therefore, $(p \after \tclass[D] t)\ \MUST\ L$
    because $\tclass[D] t \subseteq \tclass[D'] t = \tclass[D'] s$. Since
    $\mathbb{I}'$ is a refinement of $\mathbb{I}$, there is some
    $I\in\mathbb{I}$ such that $I'\subseteq I$ and $L\subseteq I$ and
    $I\in \mathbb{I}$. Consequently, $p\btmustleq[\mathbb{I}] q $ implies
    $(q \after \tclass[D] t)\ \MUST\ L$ and, hence, $(q \after t)\ \MUST\ L$
    for all $t\in \tclass [D'] s$. Therefore,
    $(q \after {\tclass[D'] s})\ \MUST\ L$.
    \qedhere
  \end{itemize}
\end{proof}

\noindent
This result allows us to conclude that the \distributed\ preorder is coarser
than the classical must testing preorder. It suffices to note that the
preorder associated to the maximal element of the partition lattice, i.e., the
trivial partition ${\mathbb{I}} =\{\actset\}$, corresponds to $\mustleq$.

\begin{lem}\label{lem:eq-must-distributed}
  ${\mustleq} = {\dmustleq{\{\actset\}}}$.
\end{lem}

\begin{proof}
  ($\subseteq$) We show that $p\bmustleq q$ implies
  $p\btmustleq[{\{\actset\}}] q$ for all $p$ and $q$. Assume
  $p\Downarrow \tclass s$ for $s\in\actset^*$. Then,
  \begin{itemize}
  \item $p\Downarrow t$ for all $t\in\tclass s$. Since $p\bmustleq q$, we know
    that $p\Downarrow t$ implies $q \Downarrow t$ for all $t\in \tclass s$.
    Consequently, $q\Downarrow \tclass s$.
  \item Assume $(p \after {\tclass s})\ \MUST\ L$ for any
    $L\subseteq \actset$. Then, $(p \after t)\ \MUST\ L$ for all
    $t\in \tclass s$. Since $p\bmustleq q $, $(q \after t)\ \MUST\ L$ for all
    $t\in \tclass s$. Hence, $(q \after \tclass s)\ \MUST\ L$.
  \end{itemize}

\noindent
($\supseteq$) We show that $p\btmustleq[{\{\actset\}}] q$ implies
$p\bmustleq q$ for all $p$ and $q$. First note that the interface
$\{\actset\}$ induces a total dependency relation $D$ on $\actset$. This
implies $\tclass[ D] s = \{s\}$ for all $s\in\actset^*$. Assume
$p\Downarrow s$ for $s\in\actset^*$. Then,
  \begin{itemize}
  \item $p\Downarrow \tclass s$. From $p\btmustleq[{\{\actset\}}] q$, we know
    that $p\Downarrow \tclass s$ implies $q \Downarrow \tclass s$.
    Consequently, $q\Downarrow s$.
  \item Assume that $(p \after {\tclass s})\ \MUST\ L$ holds for a finite
    $L\subseteq \actset$. Since $p\btmustleq[{\{\actset\}}] q $,
    $(q \after {\tclass s})\ \MUST\ L$. Hence, $(q \after s)\ \MUST\ L$.
    \qedhere
  \end{itemize}
\end{proof}

\begin{cor}\label{prop:must-implies-distributed} Let $\mathbb{I}$ be an
  interface. Then, $p \mustleq q$ implies $ p \dmustleq{\mathbb{I}} q$.
\end{cor}

The converse of Lemma~\ref{lemm-refinement-dist} and
Corollary~\ref{prop:must-implies-distributed} do not hold. Consider the
processes $p = a.b.\zero + a.\zero + b.\zero$ and
$q = b.a.\zero + a.\zero + b.\zero$. It has been shown, in
Example~\ref{ex-dmust-swap-axiom}, that we have
$p \dmustleq{{\{\{a\},\{b\}\}}} q$. Nonetheless, it is easy to check that
$p \not\mustleq q$ (i.e., $p \dmustleq{\{\actset\}} q$) by using
$o=\co{b}.(\tau.\one + \co{a}.\zero)$ as observer.

We also have that the \uncoordinated\ preorder is coarser than the
\distributed\ one.

\begin{prop}\label{lemm-refinement-unc}
  Let $\mathbb{I}$ be an interface. Then, $p\ \dmustleq {\mathbb{I}}\ q$
  implies $p \lmustleq{\mathbb{I}} q$.
\end{prop}

\begin{proof}
  Let $D$ be the dependency relation induced by $\mathbb{I}$. We first note
  that $\tclass t \subseteq \nclass s$ for all $t\in\nclass s$ and
  $I\in\mathbb{I}$ because every two strings in the same Mazurkiewicz trace
  have the same symbols and symbols in the same part of the interface do not
  commute. Then, assume $p\Downarrow \nclass s$ for $s\in\actset^*$.
  Consequently,
  \begin{itemize}
  \item $p\Downarrow t$ for all $t\in\nclass s$. Since $\nclass t = \nclass s$
    and $\tclass t \subseteq \nclass s$, $p\Downarrow {t'}$ for all
    ${t'}\in\tclass t$ and $t\in\nclass s$. Moreover,
    $p \btmustleq[{\mathbb{I}}] q$ implies $q \Downarrow {t'}$ for all
    ${t'}\in\tclass t$ and $t\in\nclass s$. Consequently,
    $q\Downarrow \nclass s$.
  \item Assume $(p \after {\nclass s})\ \MUST\ L \cup (\actset\backslash I)$
    with $L\subseteq I$. Then, $(p \abst I \after t\abst I)\ \MUST\ L$, for
    all ${t}\in\nclass s$.
    Then $p \btmustleq[{\mathbb{I}}] q$ implies $(q \abst I \after t'
    \abst I)\ \MUST\ L$ for all $t'\in \tclass t$ and $t\in\nclass s$.
    \qedhere
   \end{itemize}
\end{proof}

\noindent
The converse does not hold, i.e., $p\ \lmustleq{\mathbb{I}} q$ does not imply
$q \ \dmustleq{\mathbb{I}} \ p$. Indeed, we have that
$a.b.\zero \ \lmustleq {\{\{a\},\{b\}\}}b.a.\zero$
(Example~\ref{ab-ind-equiv-ba}) but
$a.b.\zero \ \not\dmustleq {\{\{a\},\{b\}\}} b.a.\zero$
(Example~\ref{ex-dmust-swap-axiom}).

\section{Multiparty testing at work}\label{atwork}

In this section we show how to use the proposed preorders for analysing a
larger scenario involving a replicated data store with alternative policies
for consistency. In order to support the task of checking relations, we
introduce a prototype implementation of the alternative characterisations
provided in the paper limited to the fragment of the calculus with finite
processes. We resorted to this limitation because the proposed alternative
characterisations use quantification over all possible traces of a process.
The development of decision procedures for the infinite case (e.g., along the
lines of~\cite{bernardi2017full}) is left to future work.

\subsection{Implementation in Prolog}

To provide the Prolog implementation of the new testing preorders, we rely on
their alternative characterisations in term of traces. The actual
implementation is restricted to the finite fragment of our specification
language (Sequential {\sc ccs}) and is available
from \url{https://github.com/hmelgra/Multiparty-preorders}.

Processes are represented as functional terms built-up from the constants
\PI|0| and \PI|1|, the unary operator \PI|~| (output actions), and the binary
functions \PI|*| (prefix) and \PI|+| (choice). The operational semantics of
finite {\sc ccs} processes is given by the ternary predicate \PI|red(P,L,Q)|,
which is defined in one-to-one correspondence with the inference rules for
finite processes (i.e., those rules that do not involve recursive processes)
in Definition~\ref{def:lts-processes}. The corresponding Prolog predicates are
the following.

\begin{PrologN}
red(1, tick, 0).
red(L * P, L, P).
red(P + _, L, P1) :- red(P, L, P1).
red(_ + Q, L, Q1) :- red(Q, L, Q1).
\end{PrologN}

Now, by building on the predicate \PI|red(_,_,_)|, we inductively define the
ternary (\emph{weak reduction}) relation $P \dtr{S} Q$ as the predicate
\PI|wred(P, S, Q)| below.

\begin{PrologD}
wred(P, [], P).
wred(P, [L|S], Q) :- red(P, L, R), L\=tau, wred(R, S, Q).
wred(P, S, Q) :- red(P, tau, R), wred(R, S, Q).
\end{PrologD}

The rules above respectively stand for $P \dtr{} P$ (line 1); $P \dtr{L S} Q$
if $L \neq \tau$ and $P \tr{L} R$ and $R\dtr{S}Q$ (line 2); and $P \dtr{S} Q$
if $P \tr{\tau} R$ and $R\dtr{S}Q$ (line 3).

Then, the \emph{set of traces from} $P$, $S = \trc{P}$, is defined as follows.
\begin{PrologD}
tr(P, T) :- wred(P, T, _).
str(P, S) :- setof(T, tr(P,T), S).
\end{PrologD}

The set containing the residuals of a process $P$ after the execution of a
sequence of actions $T$ is defined by the following two rules
\begin{PrologD}
after(P, T, []) :- not(wred(P, T, _)), !.
after(P, T, Qs) :- setof(Q, wred(P,T,Q), Qs).
\end{PrologD}

The first rule states that $P\after T = \emptyset$ when $P$ does not have $T$
as one of its traces, while the second one handles the case in which $T$ is a
trace of $P$. The predicate \PI|after(_,_,_)| is implemented with two rules
because \PI|setof(Q, wred(P,T,Q), Qs)| fails when the goal \PI|wred(P,T,Q)|
does not have any solution.

The predicate $P\ \MUST\ L$ of Definition~\ref{def:must-set} is inductively
implemented by the following rules.
\begin{PrologD}
must([], _).
must([P|Ps], L) :-
    member(A, L), wred(P, [A], _), !, must(Ps, L).
\end{PrologD}

Line 1 stands for the base case, i.e., $\emptyset\ \MUST\ L$ for any $L$.
Differently, Line 2 states that for a non empty set of processes
$\{P\}\cup Ps$, it should be the case that there exists some action $A\in L$
such that $P\dtr{A}$ and $Ps\ \MUST\ L$.

We have now all the ingredients needed for the definition of $\bmustleq$.
Our implementation relies on refutation, i.e., we indirectly show
$p \bmustleq q$ by falsifying $p \not\bmustleq q$. From
Definition~\ref{def-charact-dmust}, we deduce that $p\not\btmustleq q$ if
there exist $s\in\actset^*$, $I\in\mathbb{I}$, a finite $L\subseteq I$, such
that $p\Downarrow \tclass s$, and either
\begin{itemize}
\item $q\Uparrow \tclass s$, or
\item $(p \after {\tclass s})\ \MUST\ L$ and
  $(q \after {\tclass    s})\ \NOTMUST\ L$.
\end{itemize}

\noindent
Since we are considering the finite fragment of the calculus, the convergence
predicate $\Downarrow_s$ trivially holds for finite processes. Consequently,
for finite processes we have, $p\not\btmustleq q$ if there exist
$s\in\actset^*$, $I\in\mathbb{I}$, and a finite $L\subseteq I$, such that
$(p \after {\tclass s})\ \MUST\ L$ and $(q \after {\tclass s})\ \NOTMUST\ L$.
Moreover, we take advantage of the refutation procedure to obtain witnesses
that explain why two particular processes are not in $\bmustleq$ relation.
Hence, we implement $\not\bmustleq$ as the quaternary predicate
\PI|notleqmust(P,Q,S,L)| meaning that $P\not\bmustleq Q$ because
$(P \after S)\ \MUST\ L$ but $(Q \after S) \NOTMUST\ L$.

\begin{PrologD}
notleqmust(P, Q, S, L):-
    str(P+Q, Ss), member(S, Ss),
    after(P, S, Ps), after(Q, S, Qs),
    n(P+Q, As), subseteq(L, As),
    must(Ps, L), not(must(Qs, L)).

leqmust(P,Q) :- not(notleqmust(P,Q,_,_)).
\end{PrologD}

Line 2 states that we only consider the set \PI|Ss| of traces that are either
traces of \PI|P| or \PI|Q| and disregard any other trace because the residuals
for both \PI|P| and \PI|Q| are empty in those cases, and hence uninteresting.
When defining must-sets, it is useless to consider actions that are not in the
alphabet of the processes\footnote{The definition of predicate \PI|n(P,As)|,
  which computes the alphabet of \PI|P|, has been omitted because it is
  straightforward.}. Therefore, line 4 states that we only consider subsets
\PI|L| of the names occurring in either \PI|P| or \PI|Q| . Then, in order to
show that \PI|P| and \PI|Q| are not in $\bmustleq$ relation, we search for a
set \PI|L| that is a must-set of the residuals of \PI|P| after \PI|S| (i.e.,
\PI|must(Ps, L)|) but not of the residuals of \PI|Q| after \PI|S| (line 5).
Finally, the predicate $\bmustleq$ is just defined as the negation of
$\not\bmustleq$ (line 7).

As an example of use of the \PI|notleqmust(_,_,_,_)| predicate, we can use it
to show that neither $\zero \mustleq \tau.a.\zero + \tau. b.\zero$ nor
$\tau.a.\zero + \tau. b.\zero \mustleq \zero$ hold.

In fact, the following query
\begin{PrologD}
?- notleqmust(0, tau * a * 0 + tau * b * 0,S,L).
\end{PrologD}
has several solutions, among which we have \PI|S = [a], L = []|.

Similarly,
 \begin{PrologD}
?- notleqmust(tau * a * 0 + tau * b * 0, 0, S,L).
\end{PrologD}
has \PI|S = [], L = [a, b]| among its solutions.

The implementation for the \distributed\ and \uncoordinated\ preorders follows
analogously. First, we generalise the definition of residuals to consider a
set of traces instead of just a trace. This is done just by collecting all the
residuals of the process for each trace in the set. We use the ternary
predicate \PI|afterC(_,_,_)| defined as follows.

\begin{PrologD}
afterC(_,[],[]).
afterC(P,[X|Xs],Ps):- after(P,X,P1s), afterC(P,Xs,P2s),
    union(P1s,P2s,Ps).
\end{PrologD}

In addition, we use two auxiliary predicates:
\PI|independence(I,Ind)|, which computes the independence relation
\PI|Ind| induced by an interface \PI|I|; and
\PI|mazurkiewicz(Ind,S,CT)|, which takes an independence relation
\PI|Ind|, a set of traces \PI|S| belonging to the same equivalence
class, and generates the complete set of traces \PI|CT| in that
equivalence class.  We omit here their definition because are
straightforward and not interesting.  We first define the
  relation \PI|maximalPref(_,_,_,_)| for the set of maximal reductions
  \PI|MR| of a process \PI|P| within the equivalence class of \PI|T|
  (for the interface \PI|I|).
\begin{PrologD}
maximalPref(P,T,I,MR) :- independence(I,Ind), mazurkiewicz(Ind,[T],CT),
    str(P,SP), prefs(SP,CT,PSP), maximal(PSP,PSP,MR),!.
\end{PrologD}
We compute the independence relation \PI|Ind| induced by the interface \PI|I|
and the Mazurkiewicz class \PI|CT| of \PI|T|. Then, we select the maximal
elements \PI|MR| from the set of reductions \PI|SP| of \PI|P| for the class
\PI|CT| (the omitted implementations of \PI|prefs(_,_,_)| and
\PI|maximal(_,_,_)| are uninteresting). Then, the actual set of strings for
Mazurkiewicz trace is the set of maximal prefixes, if they jointly-complete
the trace (i.e., the omitted predicate \PI|dagger(_,_,_)|); otherwise is
empty.
\begin{PrologD}
strClass(P,T,I,MR):-  maximalPref(P,T,I,MR), dagger(MR,T,I).
strClass(P,T,I,[]):-  maximalPref(P,T,I,MR), not(dagger(MR,T,I)).
\end{PrologD}

As for the classical must preorder, we implement $\btmustleq$ in terms of
$\not\btmustleq$, which is defined by the predicate \PI|notlequnc(P,Q,I,T,L)|,
in which the additional parameter \PI|I| stands for the interface. Its
definition is below.

\begin{PrologD}
notlequnc(P,Q,I,T,L):-
    str(P+Q,Ts), !, member(T,Ts),
    strClass(P,T,I,CTP), strClass(Q,T,I,CTQ),
    afterC(P,CTP,Ps), afterC(Q,CTQ,Qs),
    member(PI, I), subseteq(L, PI),
    must(Ps, L), not(must(Qs,L)).

lequnc(P,Q,I) :- not(notlequnc(P,Q,I,_,_)).
\end{PrologD}

The differences with respect to the definition of \PI|notlequnc(_,_,_,_)| are
the following:
\begin{itemize}
\item we compute the set of strings with respect to an equivalence class,
  i.e., \PI|CTP| and \PI|CTQ| (line 3);
\item residuals are obtained for each equivalence class of a trace (line 4)
  (instead of just a string);
\item must-sets are built with actions in just one part of the interface (line
  5).
\end{itemize}

\noindent
Then, we can check, e.g., that
$\tau.a.\zero + \tau.b.\zero \dmustleq{\mathbb{I}} 0$ for
$\mathbb{I} = \{\{a\},\{b\}\}$ by executing the query
\begin{PrologD}
 ?- notlequnc(tau * a * 0 + tau * b * 0,0,[[a],[b]],T,L).
\end{PrologD}
which does not have any solutions.

Also, we can test that $a . b \not\dmustleq{\mathbb{I}} b.a$ for
$\mathbb{I} = \{\{a\},\{b\}\}$, because the query
\begin{PrologD}
?- notlequnc(b* a * 0 ,a * b * 0,[[a],[b]], T, L).
\end{PrologD}
has several solutions, among which we have \PI|T = [], L = [b]|.

The implementation of the \uncoordinated\ preorder consists in the definition
of an analogous predicate \PI|notleqind(P,Q,I,T,L)|, which considers the
equivalence classes of filtered traces instead of the Mazurkiewicz ones. It is
defined as follows.
\begin{PrologD}
notleqind(P,Q,I,T,L1):-
    str(P+Q,Ts), member(T,Ts), member(PI,I),
    filtered(T,PI,Ts,CT), afterC(P,CT,P1), afterC(Q,CT,Q1),
    complement(I, PI, C), subseteq(L1, PI), append(L1,C,L),
    must(P1, L), not(must(Q1,L)).

leqind(P,Q,I) :- not(notlequnc(P,Q,I,_,_)).
\end{PrologD}

In this case the variable \PI|CT| in line 3 stands for the (relevant part of
the) equivalence class of the trace \PI|T|. Since the equivalence classes for
the filtered case are all infinite and, hence, cannot be computed completely,
the predicate \PI|filtered(T,PI,Ts,CT)| simply generates the traces in the
equivalence class of \PI|T| that are also traces of at least one of the two
processes under comparison (note that the residuals are empty for both
processes in the remaining cases, and hence irrelevant). The definition of
\PI|filtered(T,PI,Ts,CT)| takes a part of the interface \PI|PI| and a set of
traces \PI|Ts|, and returns \PI|CT| which contains the set of traces in
\PI|Ts| whose projection over \PI|PI| coincides with the projection of \PI|T|.
Note that \PI|Ts| in line 3 corresponds to the traces in either \PI|P| or
\PI|Q| (line 2). The remaining difference concerns to the generation of
must-sets (line 4). In this case, each candidate must-set \PI|L| contains a
subset \PI|L1| of the part of the interface under analysis \PI|PI| and the set
\PI|C| containing all actions in the interface \PI|I| that are not in \PI|PI|
(this set is computed by the predicate \PI|complement(I, PI, C)|, whose
definition has been omitted).

We can use this predicate to check, e.g., that
$a . b.\zero \lmustleq{\mathbb{I}} b.a.\zero$ for
$\mathbb{I} = \{\{a\},\{b\}\}$ by executing the query
\begin{PrologD}
 ?- notleqind( a * b * 0, b * a * 0, [[a],[b]], T,L).
\end{PrologD}
which does not have any solution.

Also, we may check that $a.b.\zero \lmustleq{\mathbb{I}} b.a.\zero$ does not
hold when $\mathbb{I} = \{\{a, b\}\}$ because the query
\begin{PrologD}
?- notleqind( a * b * 0, b * a * 0, [[a,b]], T,L).
\end{PrologD}
has several solutions, e.g., \PI|T = [], L = [a]|.

We now illustrate the use of the introduced preorders and of our prototype
implementation in a larger scenario.

\subsection{A case study}\label{ex:dynamo}
Distributed, non-relational databases such as Dynamo~\cite{DeCandiaHJKLPSVV07}
and Cassandra~\cite{LakshmanM10} provide highly available storage by
replicating data and relaxing consistency guarantees. Such databases store
key-value pairs that can be accessed by using two operations: {\tt get} to
retrieve the value associated with a key, and {\tt put} to store the value of
a particular key. A client issuing an operation interacts with the closest
server, which plays the role of a coordinator and mediates between the client
and the replicas to complete the client request. Each client request is
associated with a consistency level, which specifies the degree of consistency
required over data. For a {\tt put} operation, the consistency level states
the number of replicas that must be written before sending an acknowledgement
to the client. Similarly, the consistency level of a {\tt get} operation
specifies the number of replicas that must reply to the read request before
returning the data to the client. Cassandra provides several consistency
levels; for instance, an operation may request to be performed over just {\tt
  ONE} or {\tt TWO} replicas, or over the majority of the replicas (\ie, {\tt
  QUORUM}) or over {\tt ALL} the replicas. Consequently, depending on the
consistency level required by the client, the coordinator chooses the replicas
to contact.

We will now describe the behaviour of a node acting as coordinator in a
configuration that involves two additional replicas. Then we will introduce
alternative policies the coordinator might want to use when reacting to users
request and will discuss their relationships.

For simplicity reasons, we just illustrate the protocol for processing the
operation {\tt get} and abstract away from the values exchanged during the
communication (the {\tt put} operation is analogous). The actual protocol for
handling a {\tt get} is described below as a {\tt CCS} process.

\bigskip

\begin{math}
  \begin{array}{r@{\ }c@{\ }l}
    {\tt Coord}
    & \bydef &
    {\tt get}.(\tau.\co{\tt err}.\zero+ \tau.\co{\tt ret}.\zero + \tau.{\tt
               Query}_1 \\ && \hspace{.8cm}
               + \tau.{\tt Query}_2 + \tau.{\tt Query}_{1,2} + \tau.{\tt Query}_{2,1})
    \\
    {\tt Query}_i
    & \bydef &
    \co{{\tt read}_i}.( \tau.\co{\tt err}.\zero
    + {{\tt ret}_i}.(\tau.\co{\tt err}.\zero
    + \tau.\co{\tt ret}.\zero))
    \\
    {\tt Query}_{i,j}
    & \bydef &
    \co{{\tt read}_i}.\co{{\tt read}_j}.(\tau.\co{\tt err}.\zero
    + {\tt Ans}_{i,j} +  {\tt Ans}_{j,i})
    \\
    {\tt Ans}_{i,j}
    & \bydef &
    {{\tt ret}_i}.(\tau.\co{\tt err} .\zero
    + \tau.\co{\tt ret}.\zero
    + {{\tt ret}_j}.(\tau.\co{\tt ret}.\zero
    +\tau.\co{\tt err}.\zero))
    \\
  \end{array}
\end{math}

\bigskip

As stated in ${\tt Coord}$, the coordinator after receiving the request
${\tt get}$ internally decides to either:
\begin{itemize}
\item reply to the client with the error message ${\tt err}$, e.g., when the
  available nodes are not enough to guarantee the requested consistency level;
  or
\item return the requested information by using just local information
  (message ${\tt ret}$); or
\item retrieve information by contacting just one of the additional replicas
  following the protocol defined by ${\tt Query}_{i}$; or
\item retrieve information from both replicas, following the protocol defined
  by ${\tt Query}_{i,j}$.
\end{itemize}

\noindent
The protocol followed by the coordinator when contacting replica $i$ is
modelled by process ${\tt Query}_i$: The coordinator sends a read request over
the channel ${{\tt read}_i}$ and awaits an answer on channel ${{\tt ret}_i}$,
however it may internally decide not to wait for the answer from the replica
and send an error to the client (e.g., in a timeout expires). When the
coordinator receives the response from the replica, it may return the
requested information to the client or signal an error (e.g., when the
consistency level cannot be satisfied by the current state of the replicas).

The protocol followed by the coordinator when contacting both replicas is
modelled by process ${\tt Query}_{i,j}$: When awaiting for their responses,
the coordinator may internally decide to reply to the client before or after
receiving any of the two answers.

Any equation ${\tt name} \bydef {proc}$ above can be defined in Prolog by
using the predicate \PI|proc(name, proc)| as shown below.

\begin{PrologD}
proc(coord,  get * (tau * ~err * 0 + tau * ~ret * 0
                     + tau * Query1    + tau * Query2
      	             + tau * Query12  + tau * Query21))
:- proc(query(1), Query1), proc(query(2), Query2),
   proc(query(1,2), Query12), proc(query(2,1), Query21).


proc(query(I), ~read(I) * (tau * ~err * 0
                              + ret(I) * (tau * ~err * 0 + tau * ~ret * 0))).

proc(query(I,J), ~read(I) * ~read(J) * (tau * ~err * 0 + AnsIJ + AnsJI))
:- proc(ans(I,J),AnsIJ), proc(ans(J,I),AnsJI).

proc(ans(I,J), ret(I) * (tau * ~ret * 0 + tau * ~err * 0
  	             + ret(J) * (tau * ~ret * 0 + tau * ~err * 0))).
\end{PrologD}

A possible implementation of ${\tt Coord}$ may only provide the part of the
protocol that always contact the two additional replicas regardless of the
information and consistency level requested by the client. Such implementation
can be described as follows,
\[
\begin{array}{rcl@{}l}
  {\tt Coord}_1
  & \bydef &
  {\tt get}. {\tt Query}_{1,2}
\end{array}
\]
where ${\tt Query}_{1,2}$ is as before.  This defining equation is
implemented in Prolog as follows,

\begin{PrologD}
proc(coord1, get * Query12) :- proc(query(1,2), Query12).
\end{PrologD}

We can check that ${\tt Coord} \mustleq {\tt Coord}_1$ by performing
the query
\begin{PrologD}
?:- proc(coord, Coord), proc(coord1,Coord1), leqmust(Coord,Coord1).
\end{PrologD}

An alternative implementation of ${\tt Coord}$ may decide to communicate an
error to the client but still accept responses from the replicas after this
interaction. This feature allows the coordinator to update its local state
with information that can be used when answering future requests. Such
implementation can be described as follows:

\bigskip

\begin{math}
  \begin{array}{rcl@{}l}
    {\tt Coord}_2
    & \bydef &
    {\tt get}.\co{{\tt read}_1}.\co{{\tt read}_2}.
    (\tau.\co{\tt err}.{{\tt ret}_1}.{{\tt ret}_2}.\zero
    + {\tt Wait}_{1,2} + {\tt Wait}_{2,1})\\
    {\tt Wait}_{i,j}
    & \bydef &
    {{\tt ret}_i}.(\tau.\co{\tt ret}.{{\tt ret}_j}.\zero
    +\tau.\co{\tt err}.{{\tt ret}_j}.\zero
    + {{\tt ret}_j}.(\tau.\co{\tt ret}.\zero
    + \tau.\co{\tt err}.\zero))
  \end{array}
\end{math}

\bigskip

Note that ${\tt Coord}_2$ accepts responses from the replicas even after it
has replied to the client. As for ${\tt Coord}$, the definition of
${\tt Coord}_2$ in Prolog is straightforward (and omitted here). When
considering the classical must testing preorder, it holds that
${\tt Coord} \not\mustleq {\tt Coord}_2$. However, as far as the behaviour of
the client and the replicas is concerned, the implementation of
${\tt Coord}_2$ is harmless. In fact, we can prove that
${\tt Coord} \dmustleq{\mathbb{I}} {\tt Coord}_2$ when
\[
  {\mathbb{I}}
  =\{
  \{{\tt get},{\tt ret},{\tt err}\},
  \{{\tt get},{\tt read}_1,{\tt ret}_1\},
  \{{\tt get},{\tt read}_2,{\tt ret}_2\}
  \}
\]

For convenience, when querying the program we add the following definition
rule for the interface.

\begin{PrologD}
int([[get,~ret,~err], [~read(1),ret(1)],[~read(2),ret(2)]]).
\end{PrologD}
and then query the program as follows:
\begin{PrologD}
?- proc(coord, C), proc(coord2,C2), int(I), notlequnc(C,C2,I,_,_).
\end{PrologD}

The query above has no solutions, and hence
${\tt Coord} \dmustleq{\mathbb{I}} {\tt Coord}_2$. Similarly, it can be proved
that ${\tt Coord}_1 \dmustleq{\mathbb{I}} {\tt Coord}_2$. On the contrary, it
can be checked the neither ${\tt Coord}_2 \dmustleq{\mathbb{I}} {\tt Coord}$
nor ${\tt Coord}_2 \dmustleq{\mathbb{I}} {\tt Coord}_1$. For instance, the
query
\begin{PrologD}
?-  proc(coord2,C2), proc(coord1, C1), int(I), notlequnc(C2,C1,I,S,L).
\end{PrologD}
has several solutions, e.g.,:

\begin{itemize}
\item
  \PI|S = [get, ~read(1), ~read(2), ret(1), ~err], L = [ret(2)]|,

\item
  \PI|S = [get, ~read(1), ~read(2), ret(2), ~err], L = [ret(1)]|,

\item
  \PI|S = [get, ~read(1), ~read(2), ~err], L = [ret(1)]|,

\item
  \PI|S = [get, ~read(1), ~read(2), ~err, ret(1)],L = [ret(2)]|.
\end{itemize}

\noindent
All of them show that ${\tt Coord}_2$ is able to accept an answer from a
replica even after signaling an error to the client, while ${\tt Coord_1}$ is
not. Consequently, a replica may distinguish the behaviours of the different
implementations: when interacting with ${\tt Coord_1}$, a replica $i$ may
discover that the coordinator has sent the message $\co{{\tt err}}$ to the
client because the interaction $\co{{\tt ret}_i}$ cannot not take place.

We now consider a variant of ${\tt Coord}_2$ that chooses a different order
for contacting replicas, defined as follows:
\[
\begin{array}{rcl@{}l}
  {\tt Coord}_3
  & \bydef &
  {\tt get}.\co{{\tt read}_2}.\co{{\tt read}_1}.
  (\tau.\co{\tt err}.{{\tt ret}_1}.{{\tt ret}_2}.\zero
  + {\tt Wait}_{1,2} + {\tt Wait}_{2,1})\\
\end{array}
\]
where ${\tt Wait}_{i,j}$ is defined as before. The only difference between
${\tt Coord}_2$ and ${\tt Coord}_3$ is the order in which $\co{{\tt read}_1}$
and $\co{{\tt read}_2}$ are executed.

We have that ${\tt Coord}_2$ and ${\tt Coord}_3$ are still distinguishable in
the \distributed\ preorder. For instance, the query
\begin{PrologD}
?- proc(coord2, C2), proc(coord3,C3), int(I), notlequnc(C2,C3,I,S,L).
\end{PrologD}
has among it solutions the following one:
\begin{PrologD}
S = [get], L = [~read(1)]
\end{PrologD}
showing that ${\tt Coord}_2 \not\dmustleq {\mathbb{I}}{\tt Coord}_3$. The test
associated with the above witness is built by preventing the interaction with
the replica $2$ (i.e., when the communication over ${\tt read}_2$ is not
enabled). However, if the interaction with the replicas is guaranteed, both
implementations should be deemed as indistinguishable. In fact,
${\tt Coord}_2$ and ${\tt Coord}_3$ are indistinguishable in the
\uncoordinated\ preorder. We remark, however, that ${\tt Coord}_1$ is still
not equivalent to either ${\tt Coord}_2$ or ${\tt Coord}_3$. For instance, the
pair
\begin{PrologD}
S = [get, ~read(1)], L = [~ret(1)]
\end{PrologD}
witnesses the fact that
${\tt Coord_3} \not\lmustleq {\mathbb{I}}{\tt Coord}_1$. In fact, while
${\tt Coord_3}$ ensures that it will always receive the reply from the replica
$1$ after sending the request ${\tt read}_1$. This is not the case for
${\tt Coord_1}$, which may refuse to communicate over ${\tt read}_1$, e.g.,
after an internal timeout.


\section{Conclusions and related work}\label{sec:conclusions}
In this paper we have explored different relaxations of the must testing
preorder tailored to define new behavioural relations that, in the framework
of Service Oriented Computing, are better suited to study compliance between
contracts exposed by clients and servers interacting via synchronous binary
communication primitives.

In particular, we have considered two different scenarios in which contexts of
a service are represented by processes with distributed control. The first
variant, that we called \distributed\ preorder, corresponds to multiparty
contexts without runtime communication between peers but with the possibility
of one peer to block another if it does not perform the expected action.
Indeed, the observations at the basis of our experiments are designed with the
assumption that the users of a service interact only via dedicated ports but
might be influenced by the fact that other partners do not perform the
expected actions. The second preorder we introduced is called \uncoordinated\
preorder. It accounts for partners that are completely independent from the
behaviour of the other ones. Indeed, from a viewpoint of a client, actions by
other clients are considered unobservable.

We have shown that the discriminating power of the induced equivalences
decreases as observers become weaker; and thus that the \uncoordinated\
preorder for a given interface is coarser than the \distributed\ preorder for
the same interface, which in turn is coarser than the classical testing
preorder. As future work we plan to consider different ``real life'' scenarios
and to assess the impact of the different assumptions at the basis of the new
preorders and the identifications/orderings they induce. We plan also to
perform further studies to get a fuller account, possibly via axiomatisations,
of their discriminating power. In the near future, we will also consider the
impact of our testing framework on calculi based on asynchronous interactions.

Several variants of the must testing preorder, contract compliance and
sub-contract relation have been developed in the literature to deal with
different aspects of services compositions, such as buffered asynchronous
communication
~\cite{bravetti2008foundational,padovani2010contract,mostrous2009global},
fairness~\cite{padovani2011fair},
peer-interaction~\cite{bernardi2013mutually}. We have however to remark that
these approaches deal with aspects that are somehow orthogonal to the
discriminating power of the distributed tests analysed in this work. Our
preorders have some similarities with those relying on buffered communications
in that both aim at guaranteeing that actions performed by independent peers
can be reordered, but we rely on synchronous communication and, hence, message
reordering is not obtained by swapping buffered messages but by relying on
more generous observers. As mentioned above, we have left the study of
distributed tests under asynchronous communication as a future work. However,
we would like to remark that, even the \distributed\ and the \uncoordinated\
preorders are different from those
in~\cite{bravetti2008foundational,padovani2010contract,mostrous2009global}
that permit explicit action reordering. The paradigmatic example is the
equivalence $a.c + b.d \lmusteq{\{a,b\},\{c,d\}} a.d + b.c$, which does not
hold for any of the preorders with buffered communication. The main reason is
that, even in presence of buffered communication, the causality, e.g., between
$a$ and $c$ is always observed.

\section*{Acknowledgments}
We are very grateful to Marzia Buscemi for interesting discussions during
early stages of this work. We thank the anonymous reviewers for their careful
reading of our manuscript and their many insightful comments and suggestions.
The work has been partially supported by CONICET (under project PIP
11220130100148CO), UBA (under project UBACYT 20020130200092BA), and MIUR under
PRIN projects CINA and IT-Matters.

\bibliographystyle{alphaurl}
\bibliography{biblio}

\appendix


\section{Proof Details of results in
  Section~\ref{sec:uncoordinated}}\label{ap:local-proofs}

In this section we provide the detailed proofs of results in
Section~\ref{sec:uncoordinated}.

\convnoisy*

\begin{proof}
\hfill
  \begin{enumerate}
  \item By contradiction. Suppose there exists $s=a_1\ldots a_n$ such that
    $p\Downarrow \nclass s$ and $q\Uparrow \nclass s$. Then, take the observer
    $o_i$ defined as follows
    \[oi = \tau.\one + \co{b^i_1}.(\tau.\one +  \ldots (\tau.\one + \co{b^i_{k_i}}.\tau.\one)\ldots)\]
    with $s_i = s\cohide {I} =b^i_1\ldots b^i_{k_i}$. Then,  $p\abst
    I\ \must\ o_i$ and $q\abst I\ \notmust\ o_i$.

    Note that $o_i \Downarrow$. Since $p\Downarrow\nclass{ s_i}$, every
    maximal computation of $(p \abst I) \| o_i $ does not diverge. For each
    maximal computation
    $(p \abst I) \| o_i \dtr{} (p' \abst I) \| o_i'\not\tr{}$ we proceed by
    unzipping the computation to conclude that $o_i \dtr{t} o_i'\not\tr{}$ for
    some $t$. It can be shown by straightforward induction on the length of
    the reduction that $o_i \dtr{t} o_i'\not\tr{}$ implies
    $o_i \dtr{t} \one \tr{\success}\zero = o'_i$. Consequently, each maximal
    computation of $(p \abst I) \| o_i $ is successful and $p\ \must\ o_i$.

    Since $q\Uparrow \nclass {s_i}$ there exists $t\in\nclass {s_i}$ such that
    $t=t_1t_2$ and $q\dtr{t_1}q'$ and $q'\Uparrow$. As before, we can conclude
    that $o_i\dtr{\co{t_1}\abst I}o_i'\dtr{\co{t_2}\abst I}$. It can be shown
    by induction on the length of $\co{t_1}$ that there exists an unsuccessful
    computation $o_i\dtr{\co{t_1}\abst I}o_i'$ such that either
    $o_i''=\tau.\one$ or $o_i''= (\tau.\one + a.\ldots)$. Then, there exists a
    maximal (divergent) unsuccessful computation
    $(q\abst I) \| o_i \dtr{} (q'\abst I) \| o''_i \tr{} (q''\abst I) \| o''_i
    \tr{} \ldots$. Consequently, $q\notmust\ o$ and this contradicts the
    assumption $p\lmustleq {\mathbb I} q$.

  \item By contradiction. Suppose there exists $s=a_1\ldots a_n$ such that
    $p\Downarrow \nclass s$, $s\in\trc q$ and for all $t\in \nclass s$,
    $t\not\in\trc p$. Then, choose $o_i$ defined as follows
    \[o_i= \tau.\one + \co{b^i_1}.(\tau.\one +  \ldots (\tau.\one + \co{b^i_{k_i}}.\zero)\ldots)  \]
    with $s_i =  s\cohide {I} =b^i_1\ldots b^i_{k_i}$.

    Note that $(q\abst I) \notmust\ o_i$ because there is a maximal
    unsuccessful computation of $(q\abst I) \| o_i$. Since $t\in \nclass s$,
    $t \abst I= t' b^i_{k_i}$. Without lost of generality, we assume
    $t= t'' b^i_{k_i}$ $t\not\in\trc p$. Then, either (i)
    $(p\abst I) \dtr{t'}\not\tr{b^i_{k_i}}$ or (ii) $(p\abst I)\not\dtr{t'}$.
    Case (i) follows immediately, because $o_i \dtr{t'} o_i'\not\tr{}$ implies
    $o_i \dtr{t'} \one \tr{\success} \zero = o'$. Hence,
    $(p\abst I) \must\ o_i$, which is in contradiction with the assumption
    that $p\lmustleq {\mathbb I} q$. Case (ii) follows by noting that
    $(p\abst I)\not\dtr{t'}$ implies that there exists $t_1$, $t_2$ with
    $t_2\neq\epsilon$ such that $p\dtr{t_1}\not\dtr{t_2}$. Moreover, if
    $o_i \dtr{t_1} o'_i\not\tr{}$ implies
    $o_i \dtr{t_2} \one \tr{\success} \zero = o'_i$. Consequently, every
    maximal computation of $(p\abst I)\| o_i$ is successful and
    $(p\abst I)\ \notmust\ o_i$, which is in contradiction with the assumption
    $p\dmustleq {\mathbb I} q$.
    \qedhere
  \end{enumerate}
\end{proof}

\mustemptynoisy*
\begin{proof}
  Suppose $\forall t\in\nclass s$, $t\not\in\trc p$. Then $(p \after
  {\nclass s}) =\emptyset$ and, by definition, $\emptyset\ \MUST\ L$
  for every finite $L\subseteq \actset$.
\end{proof}

\traceimplicationrecnoisy*
\begin{proof}
  Assume $\forall t\in\nclass s:t\not\in\trc p$. By
  Lemma~\ref{lem-aux:must-empty-noisy}, $(p \after {\nclass s})\ \MUST\ L$ for
  every finite $L\subseteq \actset$. By straightforward induction on the
  length of the reduction, we can show that $p \dtr{t} q$ implies
  $\n(q) \subseteq \n(p)$. Analogously, we can show that
  $\init q \subseteq \n(p)$. Consequently,
  $\bigcup\{\init {q'}\ |\ q\dtr{\nclass s} q'\} \subseteq \n(p)$. Since,
  $\n(p)$ is finite, we can conclude that the set
  $\bigcup\{\init{q'}\ |\ q\dtr{\nclass s} q'\}$ is finite. Therefore, we can
  find an action $a$ such that for all $t\in\nclass s$ we have
  $q\not\dtr{t a}$. Then $(q \after {\nclass s})\ \NOTMUST\ \{a\}$ while
  $(p \after {\nclass s})\ \MUST\ \{a\}$, which contradicts the hypothesis
  $p \btmustleq q$.
\end{proof}

\medskip

\thcharacterisationlocalmust*
\begin{proof}
  \hfill
  \begin{itemize}[align=left]
  \item [($\subseteq$)] Actually we prove that $p\not\bnmustleq[I] q$ implies
    $p \not\lmustleq{I} q$. Assume that there exists $s=a_1\ldots a_n$ and
    $L\subseteq I$ such that
    \begin{enumerate}
    \item $p\Downarrow \nclass s$ and $q\Uparrow \nclass s$, or
    \item $s \in \trc q$ and $\forall t\in\nclass s. t\not\in\trc p$ or
    \item $(p \after \nclass s) \ \MUST\ L \cup (\actset\backslash I)$ and
      $(q \after \nclass s) \NOTMUST\ L\cup (\actset\backslash I)$
    \end{enumerate}
    For each case we show that there exists an observer such that
    $p\abst I\ \must\ o$ and $q\abst I\ \notmust\ o$. For the two first cases,
    we take the observers as defined in proof of
    Lemma~\ref{lem-aux:conv-noisy}. For the third one, define
  \[o = \tau.\one + \co{b^i_1}.(\tau.\one + \ldots (\tau.\one +
    \co{b^i_{k_i}}.\sum_{a\in L} a.\one)\ldots)\] with
  $s\cohide {I} = b^i_1\ldots b^i_{k_i}$.
\item [($\supseteq$)] We prove $p \bnmustleq[I] q$ implies $p \lmustleq{I} q$.
  Actually, the proof follows by showing that $p \bnmustleq[I] q$ and
  $q\abst I\ \notmust\ o$ imply $p\abst I\ \notmust\ o$. Assume there exists
  an unsuccessful computation
  \[q\abst I\ \|\ o = q_0\abst I\ \|\ o_0\tr{\tau} \ldots \tr{\tau}q_k\abst I\
    \|\ o_k\tr{\tau}\ldots\] Consider the following cases:

  \begin{enumerate}
  \item\label{proof:th:comm-characterisation-item1-noisy} {\bf The computation
      is finite}, i.e., there exists $n$ such that
    $q_n\abst I\ \|\ o_n\not\tr{\tau}$ and $q_i\abst I\Downarrow$ and
    $o_i\Downarrow$ for $i\leq n$. By unzipping the computation, there exists
    $s$ such that $q_0\abst I\dtr{s} q_n\abst I$ and $o_0\dtr{\co{s}}o_n$
    unsuccessful. Note that $\n(s)\subseteq I$ and hence $s\cohide I = s$
    Moreover, $q_n \abst I \NOTMUST\ \init {o_n}$ and, hence
    $(q\abst I\ \after s) \NOTMUST\ \init{o_n}$. By
    Lemma~\ref{lem-prop-noisy}~(\ref{lem-after-noisy-class}), we have that
    $(q\ \after \ \nclass s) \NOTMUST\ \init{o_n}$.
    \begin{enumerate}
    \item Case $p \Uparrow \nclass s$. By
      Lemma~\ref{lem-prop-noisy}~(\ref{lem-div-noisy-class}),
      $p \abst I\Uparrow s$. Consequently, there is an unsuccessful
      computation of $p\abst I\ \|\ o$.
    \item Case $p\Downarrow \nclass s$. Note that $s\in\trc{q\abst I}$. By
      Lemma~\ref{lem-aux:conv-noisy}~(\ref{lem-aux:trace-implication-noisy}),
      therefore, $\exists t \in \nclass s$ and $t\in\trc {p}$. Hence,
      $(p \after \ \nclass s)\neq\emptyset$. Moreover, from
      $p \bnmustleq[I] q$ we conclude that
      $(q\ \after \ \nclass s) \NOTMUST\ \init o$ implies
      $(p\ \after \ \nclass s) \NOTMUST\ \init{o_n}$. Therefore, exists some
      $p'\in (p\ \after \ \nclass s)$ such that $p' \NOTMUST\ \init {o_n}$,
      and $p\abst I\dtr{s}p'\abst I$ is unsuccessful. Hence, there is an
      unsuccessful computation of $p\abst I\ \|\ o$.
    \end{enumerate}

  \item {\bf The computation is infinite}. We consider two cases:
    \begin{enumerate}
    \item There exists $s\in\trc{ q\abst I}$ and $\co{s}\in\trc{o}$ such that
      $q\abst I\Uparrow s$ or $o\Uparrow {\co s}$. Note that
      $\n(s)\subseteq I$ and hence $s\cohide I = s$. We proceed by case
      analysis.
      \begin{itemize}
      \item $q \abst I \Uparrow s$: By
        Lemma~\ref{lem-prop-noisy}~(\ref{lem-div-noisy-class}),
        $q \Uparrow \nclass s$. Therefore $p\Uparrow \nclass s$ because
        $p \btmustleq q$. By
        Lemma~\ref{lem-prop-noisy}~(\ref{lem-div-noisy-class}),
        $p \abst I\Uparrow s$. Therefore, there is an unsuccessful computation
        of $p\abst I\ \|\ o$.
      \item $q\abst I\Downarrow {s}$ (and $o\Uparrow {\co s}$): Therefore
        $q \Downarrow \nclass s$ by
        Lemma~\ref{lem-prop-noisy}~(\ref{lem-div-noisy-class}). Therefore
        $p \Downarrow \nclass s$ because $p \btmustleq q$. By
        Lemma~\ref{lem-aux:trace-implication-rec-noisy},
        $\exists t\in\nclass s: t\in\trc p$, hence $p\dtr{t}p'$. By
        Lemma~\ref{lem-prop-noisy}~(\ref{lem-red-noisy-class})
        $p\abst I \dtr{ t\cohide I}p'\abst I $. Note that $ t\cohide I = s$.
        Then, $p\abst I \dtr{ s }p'\abst I$. Hence, the computation obtained
        by zipping $p\abst I \dtr{ s }p'\abst I$ and $o\dtr {\co s} o'$ is
        unsuccessful.
      \end{itemize}

    \item $\forall n. (q_n \cohide I)\Downarrow$ and $o_n\Downarrow$. Take
      $s\in\actset^*$ such that $q\dtr{s}q_n$ and $q \cohide I \Downarrow s$
      (this is possible because
      $q \cohide I \ \| o \dtr{} q_n\cohide I\ \| o_n$ is unsuccessful and
      $\forall i\leq n. (q_i\cohide I )\Downarrow)$. By
      Lemma~\ref{lem-prop-noisy}~(\ref{lem-red-noisy-class}),
      $(q\cohide I)\Downarrow s$ implies
      $q\abst I\ \dtr{s\cohide I}q'\abst I$.

      By Lemma~\ref{def-charact-dmust}, $q\dtr{s}q_n$ and $q\Downarrow s$
      implies either (i) $p\Uparrow \nclass s$ or (ii) $p\Downarrow \nclass s$
      and $p\dtr{\nclass s}$.

      \begin{itemize}
      \item[(i)] $p \Uparrow \nclass s$: $\exists t\in\nclass s$ such that
        $p\Uparrow t$. By
        Corollary~\ref{aux-cor:commutative-observer-behavior}~(\ref{aux-lem:comm-observer-unsuccessful-comp}),
        $o\dtr{\co s}$ unsuccessful implies $o\dtr{\co t}$ unsuccessful, and
        hence there is an unsuccessful computation of $p\ \| o$.
      \item[(ii)] $p\Downarrow \nclass s$ and $p\dtr{\nclass s}$. Since
        $q_n\Downarrow$, in the computation there exists $q_m$ such that
        $q_n \dtr{}q_m \not\tr{}$ and $q_m\tr{a}q_{m+1}$. Moreover,
        $o_n \dtr{} o_m \tr{\co{a}} o_{m+1}$. Consequently, for all $L$ such
        that $a\not\in L$, we have $(q\ \after \ \nclass s) \NOTMUST\ L$.
        Since $p \btmustleq q$, then for all $L$ such that $a\not\in L$, we
        have $(p\ \after \ \nclass s) \NOTMUST\ L$. Therefore, there exists
        $p_n \in(p\ \after \ \nclass s) $ and $p_n\dtr{a} p_{m+1}$. Hence, for
        all $n$ there is an unsuccessful computation
        $p \| o \dtr{} p_n\| o_n \dtr{} p_{m+1} \| o_{m+1}$.
        \qedhere
      \end{itemize}
    \end{enumerate}
  \end{enumerate}
\end{itemize}
\end{proof}

\end{document}